\PassOptionsToPackage{unicode}{hyperref}
\PassOptionsToPackage{hyphens}{url}
\PassOptionsToPackage{dvipsnames,svgnames,x11names}{xcolor}
\documentclass[11pt]{article}
\usepackage[letterpaper,margin=0.85in]{geometry}

\usepackage{amsmath,amssymb}
\usepackage{iftex}
\ifPDFTeX
  \usepackage[T1]{fontenc}
  \usepackage[utf8]{inputenc}
  \usepackage{textcomp} 
\else 
  \usepackage{unicode-math}
  \defaultfontfeatures{Scale=MatchLowercase}
  \defaultfontfeatures[\rmfamily]{Ligatures=TeX,Scale=1}
\fi
\usepackage{lmodern}
\ifPDFTeX\else  
\fi
\IfFileExists{upquote.sty}{\usepackage{upquote}}{}
\IfFileExists{microtype.sty}{
  \usepackage[]{microtype}
  \UseMicrotypeSet[protrusion]{basicmath} 
}{}
\makeatletter
\@ifundefined{KOMAClassName}{
  \IfFileExists{parskip.sty}{%
    \usepackage{parskip}
  }{
    \setlength{\parindent}{0pt}
    \setlength{\parskip}{6pt plus 2pt minus 1pt}}
}{
  \KOMAoptions{parskip=half}}
\makeatother
\usepackage{xcolor}
\makeatletter
\ifx\paragraph\undefined\else
  \let\oldparagraph\paragraph
  \renewcommand{\paragraph}{
    \@ifstar
      \xxxParagraphStar
      \xxxParagraphNoStar
  }
  \newcommand{\xxxParagraphStar}[1]{\oldparagraph*{#1}\mbox{}}
  \newcommand{\xxxParagraphNoStar}[1]{\oldparagraph{#1}\mbox{}}
\fi
\ifx\subparagraph\undefined\else
  \let\oldsubparagraph\subparagraph
  \renewcommand{\subparagraph}{
    \@ifstar
      \xxxSubParagraphStar
      \xxxSubParagraphNoStar
  }
  \newcommand{\xxxSubParagraphStar}[1]{\oldsubparagraph*{#1}\mbox{}}
  \newcommand{\xxxSubParagraphNoStar}[1]{\oldsubparagraph{#1}\mbox{}}
\fi
\makeatother

\usepackage{longtable,booktabs,array}
\usepackage{calc} 
\usepackage{etoolbox}
\makeatletter
\patchcmd\longtable{\par}{\if@noskipsec\mbox{}\fi\par}{}{}
\makeatother
\IfFileExists{footnotehyper.sty}{\usepackage{footnotehyper}}{\usepackage{footnote}}
\makesavenoteenv{longtable}
\usepackage{graphicx}
\makeatletter
\def\maxwidth{\ifdim\Gin@nat@width>\linewidth\linewidth\else\Gin@nat@width\fi}
\def\maxheight{\ifdim\Gin@nat@height>\textheight\textheight\else\Gin@nat@height\fi}
\makeatother
\setkeys{Gin}{width=\maxwidth,height=\maxheight,keepaspectratio}
\makeatletter
\def\fps@figure{htbp}
\makeatother

\makeatletter
\@ifpackageloaded{caption}{}{\usepackage{caption}}
\AtBeginDocument{%
\ifdefined\contentsname
  \renewcommand*\contentsname{Table of contents}
\else
  \newcommand\contentsname{Table of contents}
\fi
\ifdefined\listfigurename
  \renewcommand*\listfigurename{List of Figures}
\else
  \newcommand\listfigurename{List of Figures}
\fi
\ifdefined\listtablename
  \renewcommand*\listtablename{List of Tables}
\else
  \newcommand\listtablename{List of Tables}
\fi
\ifdefined\figurename
  \renewcommand*\figurename{Figure}
\else
  \newcommand\figurename{Figure}
\fi
\ifdefined\tablename
  \renewcommand*\tablename{Table}
\else
  \newcommand\tablename{Table}
\fi
}
\@ifpackageloaded{float}{}{\usepackage{float}}
\floatstyle{ruled}
\@ifundefined{c@chapter}{\newfloat{codelisting}{h}{lop}}{\newfloat{codelisting}{h}{lop}[chapter]}
\floatname{codelisting}{Listing}

\makeatother
\makeatletter
\@ifpackageloaded{caption}{}{\usepackage{caption}}
\@ifpackageloaded{subcaption}{}{\usepackage{subcaption}}
\makeatother

\ifLuaTeX
  \usepackage{selnolig}  
\fi
\usepackage[]{natbib}
\usepackage{bookmark}

\usepackage{amsthm}
\usepackage{multirow}
\usepackage{mathtools}
\IfFileExists{xurl.sty}{\usepackage{xurl}}{}
\newcommand{\Normal}{\mathcal{N}}

\newtheorem{theorem}{Theorem}

\newcommand{\bmsection}{\section}

\newif\ifshowrevisions
\showrevisionsfalse
\ifshowrevisions
  \newcommand{\rev}[1]{{\bfseries #1}}
  \newenvironment{revision}{\begingroup\bfseries\boldmath}{\endgroup}
\else
  \newcommand{\rev}[1]{#1}
  
\fi

\providecommand{\harvardurl}[1]{\textbf{URL:} \url{#1}}
\renewcommand{\harvardurl}[1]{\textbf{URL:} \url{#1}}

\hypersetup{
  pdftitle={Standardized Dynamic Borrowing: Uniform Control of Information Borrowing},
  pdfauthor={Huan Wang; Fei Wu; Yeh-Fong Chen},
  pdfkeywords={Bayesian methods, dynamic borrowing, evidence synthesis, pediatric extrapolation, clinical trials, operating characteristics},
  colorlinks=true,
  linkcolor={blue},
  filecolor={Maroon},
  citecolor={Blue},
  urlcolor={Blue},
  pdfcreator={LaTeX}}

\begin{document}

\def\spacingset#1{\renewcommand{\baselinestretch}{#1}\small\normalsize}
\spacingset{1}

\title{\bf Standardized Dynamic Borrowing: Uniform Control of Information Borrowing}
\author{Huan Wang, Fei Wu, and Yeh-Fong Chen\thanks{Corresponding author: Yeh-Fong Chen, 10903 New Hampshire Avenue, Silver Spring, MD 20993, USA. Email: \texttt{yehfong.chen@fda.hhs.gov}.}\\
Division of Biometrics IX, FDA/CDER/OTS/OB, Maryland, USA}
\date{}
\maketitle

\begin{abstract}

Dynamic borrowing can improve efficiency in clinical trials using external information, but existing methods may not offer transparent, prespecified control over borrowing. We introduce Standardized Dynamic Borrowing (SDB), which combines a source-informed component with a no-borrowing component and updates their relative odds using a Bayes factor standardized by its supremum over admissible target data.  The posterior density under the no-borrowing component is obtained by normalizing the target-data likelihood. The prespecified source-component initial weight is a sharp uniform upper bound on its posterior weight, providing an interpretable borrowing constraint. When source and target models belong to the same one-parameter canonical exponential family, posterior expected local-information-ratio effective sample size equals target sample size plus borrowed source-prior effective sample size minus a nonnegative heterogeneity penalty. For normally distributed outcomes, borrowing depends on a standardized source-target conflict $Z$-score and decreases according to a standard normal kernel. We also jointly select the initial weight and target sample size to meet prespecified power and Type I error requirements. Normal and binomial illustrations demonstrate how SDB translates source-target compatibility into data-adaptive borrowing and how the calibration procedure prespecifies the source-component initial weight and target sample size under operating-characteristic constraints. These results provide a transparent framework for prespecifying, controlling, and interpreting dynamic borrowing.

\end{abstract}

\noindent%

\medskip
\spacingset{1}

\section{Introduction}

When a small target study is supported by a larger, credible source dataset, judicious borrowing can materially improve precision and decision making. The situation is common in pediatric drug development, rare diseases, and important subgroups in confirmatory trials, where the available sample size in the target population is typically limited but substantial information is available from related studies (e.g., adult programs or historical trials)~\citep{schoenfeld2009bayesian,spanakis2023addressing}. Regulatory guidance in pediatric drug development, such as ICH E11A and the EMA extrapolation reflection paper~\citep{ICH_E11A_2024,EMA_Extrapolation_2018} indicates that data from related populations may be used to reduce the extent of new studies when similarity in disease course and treatment response has been established, while it has been highlighted that the borrowed information should not be allowed to overwhelm the data from the target population~\citep{Travis2023InformativeBayes,FDA2026Bayesian}. 

As a result, there has been considerable interest in Bayesian methods that discount external information according to its compatibility with the target data, so that the amount of information taken from the source is reduced when the target data diverge from the external evidence. Such discounting strategies include robust mixture priors (RMPs), power priors, commensurate priors, hierarchical models, self-adapting mixture (SAM) priors, and elastic priors~\citep{Travis2023InformativeBayes,yang2023sam,jiang2023elastic}. The RMP uses a mixture distribution to combine a source-informed component with a proper weakly informative component and updates the source-component weight after observing the target data~\citep{schmidli2014robust,best2021assessing}. The power prior discounts the contribution from the source likelihood through an exponent~\citep{de2006power,ibrahim2015power}. The commensurate prior treats the target effect as a non-systematically biased representation of the source effect, with a commensurability parameter controlling the amount of borrowing~\citep{hobbs2012commensurate}. Hierarchical models analyze multiple related cohorts jointly by introducing cohort-specific random effects around a shared overall effect, so borrowing is driven by the estimated between-cohort heterogeneity~\citep{gamalo2019incorporating}. The SAM prior uses a data-driven mixing weight in an RMP by evaluating a likelihood ratio test against a predefined clinically significant difference in the treatment effect~\citep{yang2023sam}. Elastic priors provide another “discounting” approach by constructing a prior from the historical posterior and inflating its variance using a function of a user-specified congruence measure between historical and current data~\citep{jiang2023elastic}.

Although these “discounting” models can downweight external information when cohorts appear incompatible, the resulting amount of borrowing can remain sensitive to prespecified modeling choices. Taking the RMP as an example, the prior mixture weight is not itself a cap on the source-component posterior weight. In addition, the posterior borrowing profile depends on the selected proper weakly informative prior. Specifying this weakly informative prior can be challenging and may lack a clear scientific justification. A power-prior approach developed by Psioda and Xue provides another example~\citep{psioda2020bayesian}. They proposed choosing a power-prior exponent by measuring how well the observed target data are supported by their prior predictive distribution derived from the source. The authors proposed several approaches to specify the power-prior exponent, e.g., using a Bayesian p-value or a mapping function of the sample size ratio. However, it remains unclear whether these specifications of the power-prior exponent are suitable for defining the degree to which information is borrowed from the source data.

We propose Standardized Dynamic Borrowing (SDB), which retains the two-component posterior-mixture structure of the RMP but replaces the weakly informative posterior with the exact no-borrowing posterior.  SDB standardizes the source-to-no-borrowing Bayes factor by its supremum over the admissible target-data space. This construction guarantees that the source-component posterior weight does not exceed the prespecified source-component initial weight for any admissible target outcome. When this upper bound is attained, the corresponding target outcome is a compatibility anchor. The no-borrowing posterior is obtained by normalizing the target-data likelihood, so SDB does not require specification of a separate proper weakly informative prior.  The SDB construction also provides several useful interpretations of borrowing. 
When the source prior is obtained from a normalized source likelihood in the same one-parameter canonical exponential family as the target likelihood, the posterior expected local-information-ratio (ELIR) effective sample size (ESS) decomposition identifies the borrowed source-prior ESS as the source-prior ELIR ESS multiplied by the source-component posterior weight.  The source-component initial weight therefore bounds the fraction of the source prior's ELIR ESS that can be borrowed in the target analysis.  For the normal distribution, source--target conflict can be summarized by a standardized $Z$-score, and SDB attenuates the source-component odds according to a standard normal kernel in that score. We further formulate a method for jointly selecting the SDB source-component initial weight and target sample size using power and Type I error constraints. 

Overall, SDB preserves an RMP-like two-component posterior representation while enhancing prespecification and interpretability. These features also reflect considerations highlighted in the FDA draft guidance for Bayesian trials~\citep{FDA2026Bayesian}, including prospective justification of prior assumptions, decision criteria, and operating characteristics under plausible scenarios. 

To develop these ideas, Section 2 introduces SDB, establishes its theoretical properties, discusses sample size planning, and presents the normal and binomial special cases. Section 3 illustrates the method through a normal clinical trial case study. Section 4 provides a binomial numerical illustration. Sections 5 and 6 provide the discussion and conclusion, respectively. 

\section{Methods}

\subsection{Robust Mixture Prior}

Let $\theta\in\Theta$ denote the treatment effect in the target population, $y$ the target data, and $L(\theta;y)$ the target-data likelihood. Define source-informed and weakly informative models
\[
M_s:\ \theta\sim\pi(\theta\mid M_s),
\qquad
M_w:\ \theta\sim\pi(\theta\mid M_w),
\]
where $\pi(\theta\mid M_s)$ is a source-informed prior and $\pi(\theta\mid M_w)$ is a proper weakly informative prior. The source-informed prior reflects the belief that the source evidence can be fully borrowed in the target analysis, whereas the weakly informative prior reflects the belief that little prior information about the target effect should be assumed. The marginal likelihoods corresponding to the source-informed and weakly informative priors are 
\[
f(y\mid M_j)=\int_\Theta L(\theta;y)\pi(\theta\mid M_j)\,d\theta,
\qquad j\in\{s,w\},
\]
and the source-to-weak-component Bayes factor is
\[
\mathrm{BF}_{s,w}(y)=\frac{f(y\mid M_s)}{f(y\mid M_w)}.
\]

An RMP specifies a mixture prior
\[
\pi_{\mathrm{RMP}}(\theta)
=
w_{\mathrm{RMP}}\pi(\theta\mid M_s)
+
(1-w_{\mathrm{RMP}})\pi(\theta\mid M_w),
\qquad w_{\mathrm{RMP}}\in(0,1).
\]
Here, $w_{\mathrm{RMP}}$ is the prior probability assigned to the
source-informed model $M_s$, and $1-w_{\mathrm{RMP}}$ is the prior
probability assigned to the weakly informative model $M_w$.

Given $y$, the RMP posterior is also a two-component mixture:
\[
\pi_{\mathrm{RMP}}(\theta\mid y)
=
w_{\mathrm{RMP}}^*(y)\pi(\theta\mid M_s,y)
+
(1-w_{\mathrm{RMP}}^*(y))\pi(\theta\mid M_w,y),
\]
where
$\pi(\theta\mid M_s,y)$ and $\pi(\theta\mid M_w,y)$
denote the posterior distributions of $\theta$ under $M_s$ and $M_w$,
respectively, and
$w_\mathrm{RMP}^*(y)=\mathbb{P}(M_s\mid y)$ denotes the posterior
probability assigned to the source-informed model.

Applying Bayes' theorem, we have
\[
w_{\mathrm{RMP}}^*(y)
=
\frac{w_{\mathrm{RMP}}f(y\mid M_s)}
{w_{\mathrm{RMP}}f(y\mid M_s)+(1-w_{\mathrm{RMP}})f(y\mid M_w)}.
\]
Consequently, the odds of the source-component posterior weight are
\[
\frac{
w_\mathrm{RMP}
^*(y)
}{
1-
w
_\mathrm{RMP}
^*(y)
}
=
\frac{
w
_\mathrm{RMP}
}{
1-
w
_\mathrm{RMP}
}
\,\mathrm{BF}_{s,w}(y).
\]
Equivalently,
\[
w
_\mathrm{RMP}
^*(y)
=
\operatorname{logit}^{-1}
\left\{
\operatorname{logit}
\left(
w
_\mathrm{RMP}
\right)
+
\log \mathrm{BF}_{s,w}(y)
\right\},
\]
where $\operatorname{logit}(x)=\log\{x/(1-x)\}$ and
$\operatorname{logit}^{-1}(x)=1/(1+e^{-x})$.

Although the RMP is conceptually straightforward, it faces several practical challenges in design and analysis:
\begin{enumerate}
\item The proper weakly informative prior     $\pi(\theta\mid M_w)$ must be specified, yet its parameters (e.g., location and scale) can be difficult to justify scientifically; 
\item The RMP prior mixture weight $w_\mathrm{RMP}$ does not have a direct interpretation in terms of the realized amount of borrowing; 
\item The weakly informative posterior is generally not identical to the target-only posterior, so the RMP posterior does not interpolate exactly between source-informed borrowing and no borrowing.
\end{enumerate}

SDB addresses these limitations through a standardized Bayes factor construction that eliminates the need for a separate proper weakly informative prior and makes the source-component initial weight a sharp uniform cap on the source-component posterior weight. When the source prior is constructed by normalizing a source likelihood from the same one-parameter canonical exponential family as the target likelihood, this cap also directly bounds the borrowed source-prior information. The SDB posterior is an explicit mixture of the source-informed and exact no-borrowing posteriors.

\subsection{Standardized Dynamic Borrowing}
\label{sec:sdb-method}

At a high level, SDB can be summarized as follows. It uses a source-informed prior together with a no-borrowing reference prior. Source--target compatibility is quantified by the overlap between the source prior and the no-borrowing posterior. The observed overlap is scaled relative to the largest possible overlap, so the standardized value cannot exceed one and decreases as compatibility becomes lower. This value is then used to attenuate the source-component odds, with less borrowing as compatibility decreases. The resulting SDB posterior is a mixture of the source-informed and no-borrowing posteriors, with the prespecified source-component initial weight serving as a sharp uniform upper bound on the source-component posterior weight.

\subsubsection{Methodology}
Let $\theta\in\Theta$ denote the treatment effect in the target population, $y$ the observed target data, and $L(\theta;y)$ the target-data likelihood. We consider a source-informed model $M_s$ and a no-borrowing model $M_{nb}$.

Assume that the source and target estimands have been appropriately aligned on the same parameter scale. Under the source-informed model $M_s$, the source and target effects are modeled using the same parameter $\theta$,
corresponding to the full-borrowing assumption for this component. Assume that the source prior is a proper, bounded density on this parameter scale. When the source-informed prior is constructed directly from source data $y_s$, it may be obtained by normalizing the source-data likelihood $L_s(\theta;y_s)$, i.e.,
\[
\pi(\theta\mid M_s)
=
\frac{L_s(\theta;y_s)}
{\int_\Theta L_s(t;y_s)\,dt},
\]
where $0<\int_\Theta L_s(t;y_s)\,dt<\infty$ and the resulting normalized density is bounded.

Define
\[
I(y)=\int_\Theta L(t;y)\,dt,
\qquad
f(y\mid M_s)
=
\int_\Theta L(\theta;y)\pi(\theta\mid M_s)\,d\theta,
\]
and the admissible data space as
\[
\mathcal{Y}_{\mathrm{adm}}
=
\left\{
y:
0<I(y)<\infty,\;
0<f(y\mid M_s)<\infty
\right\}.
\]
Assume that $\mathcal Y_{\mathrm{adm}}\neq\varnothing$. For a proper, bounded source prior, this admissibility condition is satisfied for the usual nondegenerate data configurations in common one-parameter canonical exponential-family models, including the normal, binomial, Poisson, and exponential cases, although boundary or degenerate data configurations may require separate treatment. We restrict attention to admissible target-data values so that the source-informed and no-borrowing posteriors are well defined; together with the assumed boundedness of the source prior, these conditions ensure that the subsequent quantities required for the SDB construction are well defined.

For $y\in\mathcal Y_{\mathrm{adm}}$, the source-informed posterior is
\[
\pi(\theta\mid M_s,y)
=
\frac{L(\theta;y)\pi(\theta\mid M_s)}
{f(y\mid M_s)}.
\]

The no-borrowing model uses a flat reference prior $\pi(\theta\mid M_{nb})=c>0$ on the declared parameter scale, where $c$ is arbitrary. This prior may be improper when $\Theta$ is unbounded.  
For $y\in\mathcal Y_{\mathrm{adm}}$, the no-borrowing
posterior is the normalized target likelihood,
\[
\pi(\theta\mid M_{nb},y)
=
\frac{L(\theta;y)}{I(y)},
\]
and the corresponding reference marginal quantity is
\[
f(y\mid M_{nb})=cI(y).
\]

To quantify source--target compatibility, define the overlap function
\[
O(y)=\int_\Theta \pi(\theta\mid M_{nb},y)\pi(\theta\mid M_s)\,d\theta.
\]
The overlap function is the zero-lag cross-correlation of the source-prior and target-only posterior densities and measures how strongly the two density curves align on the declared parameter scale.
Because
\begin{align*}
f(y\mid M_s)
&=
\int_\Theta L(\theta;y)\pi(\theta\mid M_s)\,d\theta\\
&=
I(y)
\int_\Theta
\pi(\theta\mid M_{nb},y)
\pi(\theta\mid M_s)\,d\theta\\
&=
I(y)O(y)\\
&=
\frac{1}{c}f(y\mid M_{nb})O(y),
\end{align*}
the raw source-to-reference Bayes-factor quantity is
\[
\mathrm{BF}(y)
=
\frac{f(y\mid M_s)}{f(y\mid M_{nb})}
=
\frac{1}{c}O(y).
\]
Thus, $O(y)$ and $\mathrm{BF}(y)$ are proportional over $\mathcal Y_{\mathrm{adm}}$.

Because the source prior is bounded and $\pi(\theta\mid M_{nb},y)$ is a probability density for $y\in\mathcal Y_{\mathrm{adm}}$,
\[
O(y)
\leq
\sup_{\theta\in\Theta}\pi(\theta\mid M_s)
<\infty.
\]
Therefore,
\[
O_{\max}
:=
\sup_{u\in\mathcal Y_{\mathrm{adm}}}O(u)
\]
is finite and strictly positive. Define the standardized Bayes factor by
$$
\mathrm{SBF}(y)
:=
\frac{O(y)}{O_{\max}}.
$$
For any fixed choice of the arbitrary scaling constant \(c>0\) for the no-borrowing reference prior, this ratio is equivalently
$$
\mathrm{SBF}(y)
=
\frac{\mathrm{BF}(y)}
{\sup_{u\in\mathcal Y_{\mathrm{adm}}}\mathrm{BF}(u)},
$$
so the arbitrary scaling constant cancels under standardization. It follows immediately that
$$
0<\mathrm{SBF}(y)\leq1,
\qquad
\sup_{y\in\mathcal Y_{\mathrm{adm}}}\mathrm{SBF}(y)=1.
$$
If the supremum $O_{\max}$ is attained in $\mathcal Y_{\mathrm{adm}}$, any
\[
y_c
\in
\arg\max_{y\in\mathcal Y_{\mathrm{adm}}}O(y)
\]
may be designated prospectively as a compatibility anchor. In this case, $O_{\max}=O(y_c)$ and $\mathrm{SBF}(y_c)=1$. If the maximizer is not unique, all maximizers have the same standardized Bayes factor.

Let $w\in(0,1)$ denote the SDB source-component initial weight. The corresponding source-component posterior weight $w^*(y)$ is defined by updating the initial odds with the standardized Bayes factor:
\[
\frac{w^*(y)}{1-w^*(y)}
=
\frac{w}{1-w}\mathrm{SBF}(y),
\]
or equivalently,
\[
w^*(y)
=
\operatorname{logit}^{-1}
\left\{
\operatorname{logit}(w)+\log\mathrm{SBF}(y)
\right\}.
\]
The resulting SDB posterior is
\[
\pi_{\mathrm{SDB}}(\theta\mid y)
=
w^*(y)\pi(\theta\mid M_s,y)
+
(1-w^*(y))\pi(\theta\mid M_{nb},y).
\]

Because $\mathrm{SBF}(y)\leq1$,
\[
w^*(y)\leq w
\quad\text{for all }y\in\mathcal Y_{\mathrm{adm}},
\]
and, because $\sup_{y\in\mathcal Y_{\mathrm{adm}}}\mathrm{SBF}(y)=1$,
\[
\sup_{y\in\mathcal Y_{\mathrm{adm}}}w^*(y)=w.
\]
Therefore, the SDB source-component initial weight is a sharp uniform upper bound on the source-component posterior weight. If a compatibility anchor $y_c$ exists, the bound is attained and
\[
w^*(y_c)=w.
\]

When the flat reference prior can be normalized to a proper prior, with $c$ chosen as its normalizing constant,  $w^*(y)$ is an ordinary posterior model probability after calibration of a prior model weight. Specifically, define
$$\frac{\tilde w}{1-\tilde w}
:=
\frac{w}{1-w}
\frac{1}{
\sup_{u\in\mathcal Y_{\mathrm{adm}}}\mathrm{BF}(u)},
\qquad\text{i.e.,}\quad
\tilde w
:=
\frac{w}{
w+(1-w)\sup_{u\in\mathcal Y_{\mathrm{adm}}}\mathrm{BF}(u)},$$
then
\[
\frac{w^*(y)}{1-w^*(y)}
=
\frac{\tilde w}{1-\tilde w}\mathrm{BF}(y)
=
\frac{w}{1-w}\mathrm{SBF}(y).
\]
Therefore, $w^*(y)=\mathbb{P}(M_s\mid y,\tilde w)$ under a proper flat reference prior with calibrated prior model weight $\tilde w$. 
When the flat reference prior is improper, $w^*(y)$ is obtained as the limit of posterior source-model probabilities under increasingly diffuse proper reference priors with calibrated prior weights (see the
Appendix).

\subsubsection{Uniform Control of Information Borrowing}

To demonstrate how SDB facilitates transparent, prespecified control of the fraction of source-prior information that can be borrowed, we employ the ELIR, which is a method to quantify the information for the parameter of a statistical model in terms of an equivalent ESS~\citep{neuenschwander2020predictively}. Several alternative ESS methods have been proposed, including the standard conjugate ESS (e.g., $a + b$ for a $\mathrm{Beta}(a,b)$), variance- or precision-ratio ESS, and the Morita-Thall-Müller approach~\citep{PennelloThompson2008,Morita2008ESS,neuenschwander2020predictively}.  We focus on ELIR because its curvature-based definition can be applied directly to mixture distributions~\citep{neuenschwander2020predictively}, facilitating decomposition of the SDB posterior. 

Let $i_1(\theta)$ denote the expected Fisher information for a single observation under the target likelihood. For a positive, twice-differentiable density $g(\theta)$ for which the expectation below exists, define the log-density curvature as:
\begin{align*}
i_g(\theta) := -\frac{d^2}{d\theta^2} \log g(\theta).
\end{align*}
Then the ELIR of $g$ is defined as:
\begin{align*}
\mathrm{ELIR}\{g\} := \mathbb{E}_{\theta \sim g} \left[ \frac{i_g(\theta)}{i_1(\theta)} \right].
\end{align*}
The following theorem provides a decomposition of the posterior ELIR for SDB on the canonical parameter scale under the exponential family assumption. 

\begin{theorem}[SDB Posterior ELIR ESS Decomposition]\label{thm:sdb-elir}
Assume that the source and target likelihoods belong to the same regular one-parameter canonical exponential family, with source sample size $n_s$ and fixed dispersion parameter $\phi_s$, target sample size $n$ and fixed dispersion parameter $\phi$, and a common log-partition function $\psi$ and canonical parameter $\theta$.
Also assume that the source likelihood is integrable over the canonical parameter space and that the source prior is obtained by normalizing it. In addition, assume that the SDB admissible data space is nonempty and that the expectations below are finite. For any admissible $y$, let $s_s(\theta; y) := \frac{d}{d\theta} \log \pi(\theta \mid M_s, y)$ and $s_{nb}(\theta; y) := \frac{d}{d\theta} \log \pi(\theta \mid M_{nb}, y)$ denote the score functions of the source and no-borrowing component posteriors, and define the responsibility function of the source component in the posterior mixture by
\[
r_s(\theta;y)=\frac{w^*(y)\pi(\theta\mid M_s,y)}{\pi_{\mathrm{SDB}}(\theta\mid y)}.
\]
Then
\[
\mathrm{ELIR}\{\pi_{\mathrm{SDB}}(\cdot\mid y)\}
=
n+w^*(y)\mathrm{ELIR}\{\pi(\cdot\mid M_s)\}-\gamma(y),
\]
where
\[
\gamma(y)=
\mathbb E_{\theta\sim\pi_{\mathrm{SDB}}(\cdot\mid y)}
\left[
\frac{r_s(\theta;y)(1-r_s(\theta;y))
(s_s(\theta;y)-s_{nb}(\theta;y))^2}{i_1(\theta)}
\right]\geq0.
\]
\end{theorem}

\begin{proof}
See the Appendix.
\end{proof}

The term $w^*(y)\mathrm{ELIR}\{\pi(\cdot\mid M_s)\}$ is the borrowed source-prior ESS, obtained by multiplying the source-prior ELIR ESS by the source-component posterior weight. Here and below, a borrowed ESS term refers to the weighted ELIR ESS carried by a prior component before the heterogeneity penalty is subtracted. 
Because $w^*(y)\leq w$, $w^*(y)\mathrm{ELIR}\{\pi(\cdot\mid M_s)\}\leq  w\mathrm{ELIR}\{\pi(\cdot\mid M_s)\}$. Therefore, $w$ directly bounds the fraction of the source prior's ELIR ESS that can be borrowed. 

The term $\gamma(y)$ is a separate nonnegative heterogeneity penalty inherent to the ELIR for mixture distributions. It accounts for the information loss arising from mixing distinct component posteriors, so that the posterior ELIR ESS is not simply the weighted sum of the component ELIR ESS values. This penalty is zero when the component posteriors coincide, and its pointwise contribution is attenuated in regions where either component dominates the posterior mixture. 

Consequently, the posterior ELIR ESS decomposes transparently as:
\begin{center}
\textbf{SDB Posterior ESS = Target Sample Size + Borrowed Source-Prior ESS - Heterogeneity Penalty}.
\end{center}
This decomposition provides a clear understanding of information borrowing while fully accounting for posterior uncertainty.

For comparison, suppose that, in addition to the conditions of Theorem~\ref{thm:sdb-elir}, the RMP weakly informative prior can be represented, on the same canonical parameter scale, as a proper normalized likelihood contribution from hypothetical data in the same one-parameter canonical exponential family.  Under this condition,
\[
\mathrm{ELIR}\{\pi(\cdot\mid M_w,y)\}
=
n+\mathrm{ELIR}\{\pi(\cdot\mid M_w)\}.
\]
The same mixture-curvature identity used in the proof of Theorem~\ref{thm:sdb-elir} therefore gives
\begin{align*}
\mathrm{ELIR}\{\pi_{\mathrm{RMP}}(\cdot\mid y)\}
=&
n
+
w_{\mathrm{RMP}}^*(y)
\mathrm{ELIR}\{\pi(\cdot\mid M_s)\}
+
(1-w_{\mathrm{RMP}}^*(y))
\mathrm{ELIR}\{\pi(\cdot\mid M_w)\}\\
&-
\gamma_{\mathrm{RMP}}(y),
\end{align*}
where $\gamma_{\mathrm{RMP}}(y)\geq0$ is the analogous mixture heterogeneity penalty. Consequently,
\begin{center}
\textbf{RMP Posterior ESS = Target Sample Size + Borrowed Source-Prior ESS + Borrowed Weakly-informative-Prior ESS - Heterogeneity Penalty}.
\end{center}
Here, $w_{\mathrm{RMP}}^*(y)\mathrm{ELIR}\{\pi(\cdot\mid M_s)\}$ is the borrowed source-prior ESS, whereas $(1-w_{\mathrm{RMP}}^*(y))
\mathrm{ELIR}\{\pi(\cdot\mid M_w)\}$ is the borrowed weakly-informative-prior ESS. Unlike the SDB decomposition, the RMP decomposition contains an additional ELIR contribution from the weakly informative prior. Thus, the posterior information depends on both the source-prior and weakly-informative-prior information contributions.

\subsection{Power, Type I Error, and Sample Size Planning}
\label{sec:joint-design}

For a fixed SDB source-component initial weight $w$, let $Y$ denote the random target data of size $n$, $T_n=T_n(Y)$ the target-data summary entering the decision rule, and $\mathcal R_n(w)$ the set of values of $T_n$ for which SDB declares success. For a fixed true target effect $\theta^\dagger$, define the power function
\[
\rho(n,w,\theta^\dagger)
=
\mathbb P_{Y\mid\theta^\dagger,n}
\{T_n(Y)\in\mathcal R_n(w)\}.
\]
The true effect $\theta^\dagger$ is used only to evaluate the operating characteristics; the prespecified SDB analysis
continues to treat the target effect as unknown. For a one-sided null space $\Theta_0$, the frequentist Type I error is
\[
\alpha(n,w)
=
\sup_{\theta^\dagger\in\Theta_0}\rho(n,w,\theta^\dagger).
\]
The following theorem shows that, for regular one-parameter canonical exponential-family
models, the supremum of the Type I error rate is attained at the boundary of a one-sided
null parameter space.

\begin{theorem}[Monotone power and boundary attainment]
\label{thm:monotone-power}
Fix $n$ and $w$, and let
\[
\Theta_0=\{\theta^\dagger\in\Theta:\theta^\dagger\geq\theta_0\}.
\]
Suppose that, for admissible $y$, SDB declares success when
\[
\mathbb P_{\mathrm{SDB}}(\theta<\theta_0\mid y)
\geq1-\alpha_0,
\]
and that outcomes outside $\mathcal Y_{\mathrm{adm}}$ are handled by a prespecified convention preserving a monotone rejection region.
Assume that the target-data model is a regular one-parameter canonical exponential family with a scalar sufficient statistic whose support does not depend on the parameter. Then $\rho(n,w,\theta^\dagger)$ is non-increasing in $\theta^\dagger$, and
\[
\alpha(n,w)=\rho(n,w,\theta_0).
\]
\end{theorem}

\begin{proof}
See the Appendix.
\end{proof}

By reversing the inequalities in Theorem~\ref{thm:monotone-power}, the
analogous conclusions hold for the null space
$\{\theta^\dagger\in\Theta:\theta^\dagger\leq\theta_0\}$ and the
upper-tail success rule
$\mathbb P_{\mathrm{SDB}}(\theta>\theta_0\mid y)\geq1-\alpha_0$.

Under the same canonical exponential-family assumptions and corresponding one-sided decision rule, the monotonicity and boundary-attainment conclusions also apply to an RMP with a fixed prespecified prior mixture, by the monotone likelihood-ratio property of the target-data model. 

The FDA draft Bayesian guidance notes that sample size, prior distributions, success criteria, and other design elements may be adjusted to obtain desired operating characteristics~\citep{FDA2026Bayesian}. Motivated by this principle, we jointly select the design sample size $n_{\mathrm{des}}$ and source-component initial weight
$w_{\mathrm{des}}$ to satisfy prespecified power and Type~I error criteria.

For a planning alternative $\theta_A$, required power $1-\beta$, and maximum acceptable Type I error $\alpha_{\max}$, define, for each source-component initial weight $w\in(0,1)$,
\[
n(w)
=
\min\left\{
n\in\mathbb N:
\rho(n,w,\theta_A)\geq1-\beta
\ \text{and}\
\alpha(n,w)\leq\alpha_{\max}
\right\},
\]
with $n(w)=\infty$ if no sample size satisfies both criteria. 

If $n(w)<\infty$ for at least one $w\in(0,1)$, define the design sample size by
\[
n_{\mathrm{des}}
=
\min_{w\in(0,1)} n(w),
\]
and the set of design weights achieving this minimum by
\[
\mathcal W_{\mathrm{des}}
=
\arg\min_{w\in(0,1)} n(w).
\]
A design weight $w_{\mathrm{des}}\in\mathcal W_{\mathrm{des}}$ may then be selected according to a prespecified tie-breaking criterion if the minimizer is not unique. 

In practice, the minimization can be implemented using a prespecified numerical search over the initial weight $w$, and a pair attaining the smallest acceptable sample size can be selected as $(n_{\mathrm{des}},w_{\mathrm{des}})$.

An RMP can be designed using the same operating-characteristic criteria. However, the SDB initial weight $w$ has a direct borrowing interpretation because $w$ specifies the sharp uniform upper bound on the source-component posterior weight. In contrast, the RMP prior weight $w_{\mathrm{RMP}}$ does not have the same posterior-scale interpretation, and its relationship to the realized source-component posterior weight depends on the specification of the weakly informative prior; recent work has similarly highlighted that posterior borrowing in an RMP depends jointly on the prior mixture weight and the weakly informative prior specification~\citep{weru2026information,ratta2026interplay}. Thus, SDB allows the borrowing–operating-characteristic tradeoff to be prespecified and evaluated on a more interpretable scale through its initial weight.

\subsection{Normal Distribution Case}
\label{sec:normal-case}

Consider the case where both source and target data follow normal distributions. Let the target data $y$ consist of $n$ observations with sample mean $\hat{\theta}$ and fixed observation-level variance $\sigma^2$. Assume $\theta\mid M_s \sim \Normal(\mu_s,\tau_s^{-1})$. Define the target precision as $\tau = n/\sigma^2$ and the conflict measure as $\Delta = \hat{\theta} - \mu_s$.

The key results for SDB in the normal case are:

\begin{enumerate}

\item \textbf{Component Posteriors:}
The source-component posterior follows 
\begin{align*}
\theta \mid M_{s},\hat{\theta} &\sim \Normal\left(\mu_s + \frac{\tau}{\tau_s + \tau}\Delta, (\tau_s + \tau)^{-1}\right).
\end{align*}
The no-borrowing component posterior follows
\begin{align*}
\theta \mid M_{nb},\hat{\theta} &\sim\Normal\left(\hat{\theta} , \tau^{-1}\right).
\end{align*}

\item \textbf{Compatibility Anchor and Posterior Weight:}
The compatibility anchor is $\hat{\theta}_c=\mu_s$ (equivalently $\Delta=0$). The source-component posterior weight is 
\begin{align*}
w^*(\hat{\theta}) = \operatorname{logit}^{-1}\left\{\operatorname{logit}(w) - \frac{\tau_s \tau}{2(\tau_s + \tau)}\Delta^2\right\}.
\end{align*}

\item  \textbf{Geometric Interpretation of Borrowing Decay:}
Let $Z$ denote the conflict between the target estimate and the source prior, standardized by the standard deviation of the prior predictive distribution of $\hat{\theta}$ under the source model $M_s$. This statistic follows a standard normal distribution:
\begin{align*}
Z := \frac{\hat{\theta} - \mu_s}{\sqrt{\operatorname{Var}_{M_s}(\hat{\theta})}} = \frac{\Delta}{\sqrt{\tau^{-1} + \tau_s^{-1}}} \sim \Normal(0,1).
\end{align*}
The SDB update from initial source-component odds to posterior source-component odds satisfies the identity:
\begin{align*}
\frac{w^*(\hat{\theta}) /(1-w^*(\hat{\theta}) )}{w/(1-w)}
=
\mathrm{SBF}(\hat{\theta})
=
\exp\left\{-\frac{1}{2}Z^2\right\}.
\end{align*}
This shows that, when source--target conflict is standardized on the source prior predictive scale, SDB attenuates the source-component odds according to a standard normal kernel in the standardized conflict (Figure~\ref{fig:odds_decay}).

\begin{figure}[H]
    \centering
    \includegraphics[width=0.8\linewidth]{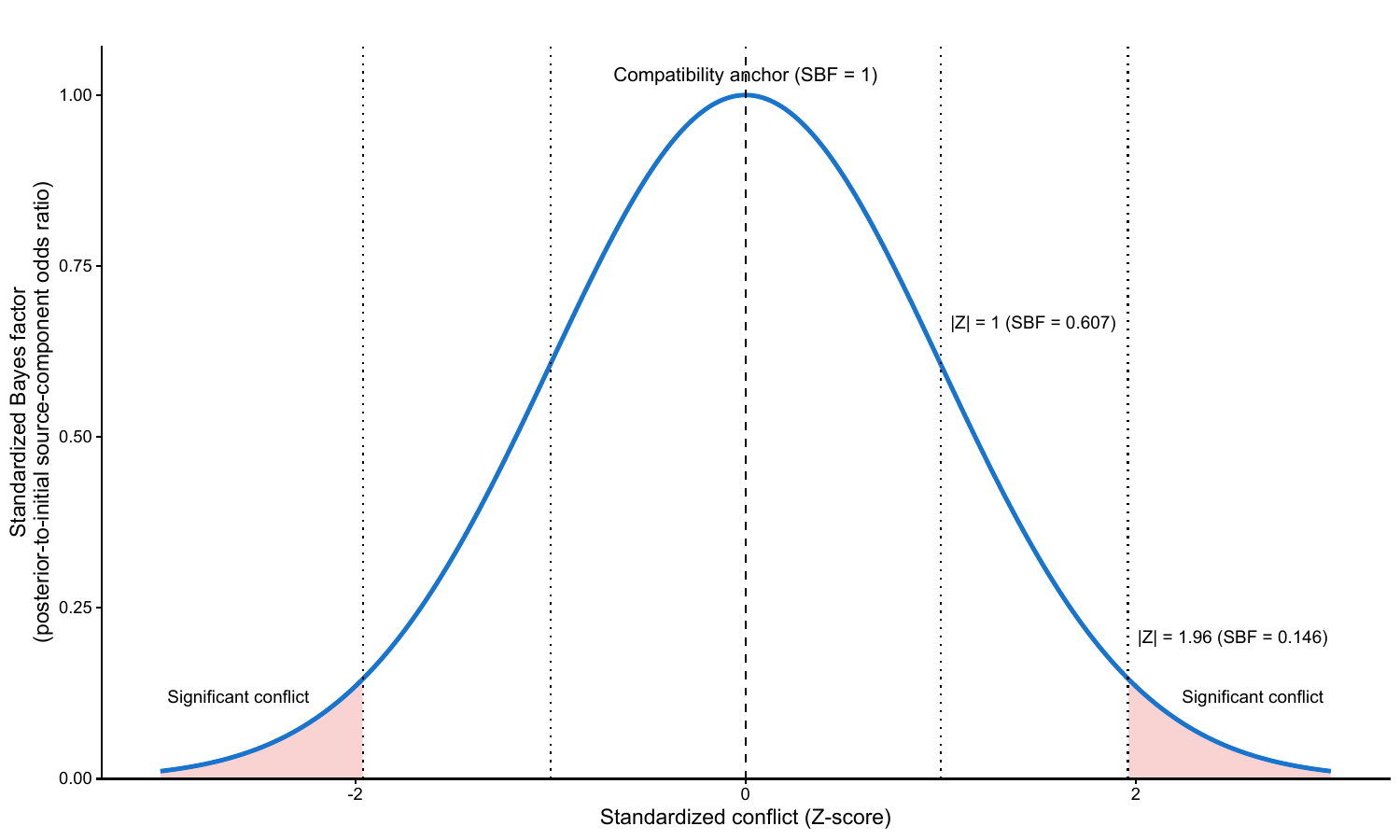}
    \caption{In the normal model, the posterior-to-initial source-component odds ratio in SDB follows a standard normal kernel in the standardized conflict.}
    \label{fig:odds_decay}
\end{figure}

\item \textbf{Posterior Decision Threshold:}
Consider the posterior decision threshold $1-\alpha_0$ used in the SDB success criterion. Under the no-borrowing analysis, for a lower-tail decision rule with null boundary $\theta_0$, the no-borrowing posterior satisfies
\[
\mathbb P(\theta<\theta_0\mid M_{nb},\hat{\theta})
=
\Phi\left\{\sqrt{\tau}(\theta_0-\hat{\theta})\right\},
\]
where $\Phi$ denotes the standard normal distribution function. Therefore,
\[
\mathbb P(\theta<\theta_0\mid M_{nb},\hat{\theta})
\geq1-\alpha_0
\]
is exactly equivalent to
\[
\sqrt{\tau}(\hat{\theta}-\theta_0)
\leq
\Phi^{-1}(\alpha_0).
\]
At the null boundary $\theta=\theta_0$, $\sqrt{\tau}(\hat{\theta}-\theta_0)\sim\Normal(0,1)$, and hence
\begin{align*}
\mathbb P_{\hat{\theta}\mid\theta_0,n}
\left\{
\mathbb P(\theta<\theta_0\mid M_{nb},\hat{\theta})
\geq1-\alpha_0
\right\}
=
\mathbb P_{\hat{\theta}\mid\theta_0,n}
\left\{
\sqrt{\tau}(\hat{\theta}-\theta_0)
\leq\Phi^{-1}(\alpha_0)
\right\}
=
\alpha_0.
\end{align*}
Thus, under no borrowing, the decision rule induced by the posterior threshold $1-\alpha_0$ has an exact one-sided frequentist Type~I error of $\alpha_0$. The analogous upper-tail no-borrowing rule has the same exact property.                 

\end{enumerate}

\subsection{Binomial Distribution Case}
\label{sec:binomial-case}

Consider a binomial target likelihood and a Beta source-informed prior on a common success-probability scale. Let $p\in(0,1)$ denote the target success probability and let
\[
\theta=\operatorname{logit}(p),
\qquad
p=\operatorname{logit}^{-1}(\theta),
\qquad
\theta\in\mathbb{R},
\]
so that $\theta$ is the canonical parameter.

Assume that the source-informed prior on the probability scale is
\[
p\mid M_s
\sim
\mathrm{Beta}(\alpha_s,\beta_s),
\qquad
\alpha_s>0,\quad\beta_s>0.
\]
The normalized-source-likelihood construction based on $x_s$ successes among $n_s$ source subjects, with $0<x_s<n_s$, is obtained as the special case
\[
\alpha_s=x_s,
\qquad
\beta_s=n_s-x_s.
\]
For the target data, let $X=x$ denote the number of observed successes among $n$ subjects, where
\[
X\mid p\sim\mathrm{Binomial}(n,p).
\]
Write $B(\cdot,\cdot)$ for the Beta function. For an interior target count $0<x<n$,
\[
\int_{\mathbb{R}}L(\theta;x)\,d\theta
=
\binom{n}{x}B(x,n-x)
<\infty,
\]
so the target likelihood can be normalized and the no-borrowing component posterior is proper. For $n\geq2$, the following SDB results apply to interior target counts $x\in\{1,\ldots,n-1\}$. Boundary counts are treated using the decision convention described below.

\begin{enumerate}
\item \textbf{Component Posteriors:}
On the probability scale, 
\begin{align*}
&p\mid M_s,x \sim \mathrm{Beta}(\alpha_s + x, \beta_s + n - x), \\
&p\mid M_{nb},x \sim \mathrm{Beta}(x, n - x).
\end{align*}

\item \textbf{Compatibility Anchor and Posterior Weight:}
If $(n-1)\tfrac{\alpha_s}{\alpha_s+\beta_s}$ is not an integer, 
\[
x_c
=
\left\lceil
(n-1)\frac{\alpha_s}{\alpha_s+\beta_s}
\right\rceil
\]
is the unique compatibility anchor; otherwise, there are exactly two adjacent
compatibility anchors,
\[
(n-1)\dfrac{\alpha_s}{\alpha_s+\beta_s}
\qquad\text{and}\qquad
(n-1)\dfrac{\alpha_s}{\alpha_s+\beta_s}+1,
\]
and either may be selected as $x_c$.

The source-component posterior weight is
\[
w^*(x)
=
\operatorname{logit}^{-1}
\left\{
\operatorname{logit}(w)+\log\mathrm{SBF}(x)
\right\},
\]
where
\[
\mathrm{SBF}(x)
=
\frac{B(x+\alpha_s,n-x+\beta_s)}
{B(x,n-x)}
\times
\frac{B(x_c,n-x_c)}
{B(x_c+\alpha_s,n-x_c+\beta_s)}.
\]

\item \textbf{Posterior Decision Threshold:}
Consider the posterior decision threshold $1-\alpha_0$ used in the SDB success criterion. Under the no-borrowing analysis, for an upper-tail decision rule with null boundary $0<p_0<1$, with the event at the boundary counts
interpreted according to the convention below,
\[
\mathbb P_{X\mid p_0,n}
\left\{
\mathbb P(p>p_0\mid M_{nb},X)
\geq1-\alpha_0
\right\}
\longrightarrow
\alpha_0
\qquad\text{as }n\rightarrow\infty.
\]
Thus, under no borrowing, the decision rule induced by the posterior
threshold $1-\alpha_0$ has an asymptotic one-sided frequentist Type~I
error of $\alpha_0$. The analogous result holds for the lower-tail decision rule.

\end{enumerate}

For boundary counts, a monotone convention can be used for the decision rule: the boundary outcome in the direction of benefit is included in the rejection region, whereas the opposite boundary is not. Thus, if a larger success probability is favorable, $x=n$ is classified as a positive result and $x=0$ is not; the assignments are reversed if a smaller success probability is favorable. This convention extends only the decision rule to the full binomial sampling support. The SDB component-posterior, overlap, standardized Bayes factor, and posterior-weight formulas remain restricted to interior counts $0<x<n$.

\section{Example: Redesign of the Best et al. Case Study}

\subsection{Design Inputs and Sample Size Planning}

We redesign the severe-asthma adolescent trial using the example presented by Best et al.~\citep{best2021assessing}. On the log rate-ratio scale, Best et al.\
used an adult-informed source prior of
\[
\theta\mid M_s
\sim
\Normal(-0.694,0.017),
\]
and an RMP with weakly informative prior $\Normal(0,12.4)$, with the variance corresponding to one unit of target information.

For illustration, assume that the adolescent unit-information variance of $12.4$ can be specified accurately at the design stage.
Efficacy is declared when the one-sided 97.5\% posterior upper bound is below zero. In this normal known-variance setting, as shown in Section~\ref{sec:normal-case}, the SDB decision rule based on this posterior threshold has a frequentist Type I error of $0.025$ under no borrowing.
For SDB, the compatibility anchor is
\[
\hat{\theta}_c
=
\mu_s
=
-0.694.
\]
The sample size calculation evaluates power under the maximum source–target compatibility scenario $\theta^\dagger=-0.694$. For the design illustration, suppose that the trial sponsor and the regulatory agency agree that at least 90\% power and a one-sided frequentist Type I error no greater than $\alpha_{\max}=0.05$ are appropriate for the adolescent trial. In this normal setting, the smallest possible integer sample size satisfying the operating-characteristic requirements is $n_{\mathrm{des}}=221$, corresponding to $w_{\mathrm{des}}=0.649$ using the minimum-weight tie-breaker: among the SDB initial weights attaining $n_{\mathrm{des}}$, the smallest weight is selected. 

To illustrate how the minimum sample size $n_{\mathrm{des}}=221$ is determined by the joint power and Type I error constraints, Figure~\ref{fig:sdb-nw} shows the sample-size requirement $n(w)$ as a function of the SDB source-component initial weight $w$. The shape reflects two competing constraints. For smaller $w$, power is limiting: increasing the allowable borrowing shifts the efficacy boundary in the favorable direction, increasing power and reducing the required target sample size. Near $w=0.65$, the design reaches the theoretical minimum $n=221$. Beyond this region, Type I error becomes limiting: additional borrowing can no longer reduce the sample size; instead, more target information is required so that, under the null hypothesis, the target data can sufficiently attenuate the favorable source contribution and maintain Type I error control. Consequently, the sample-size curve increases at larger values of $w$.

\begin{figure}[!htbp]
\centering
\includegraphics[width=0.8\linewidth]{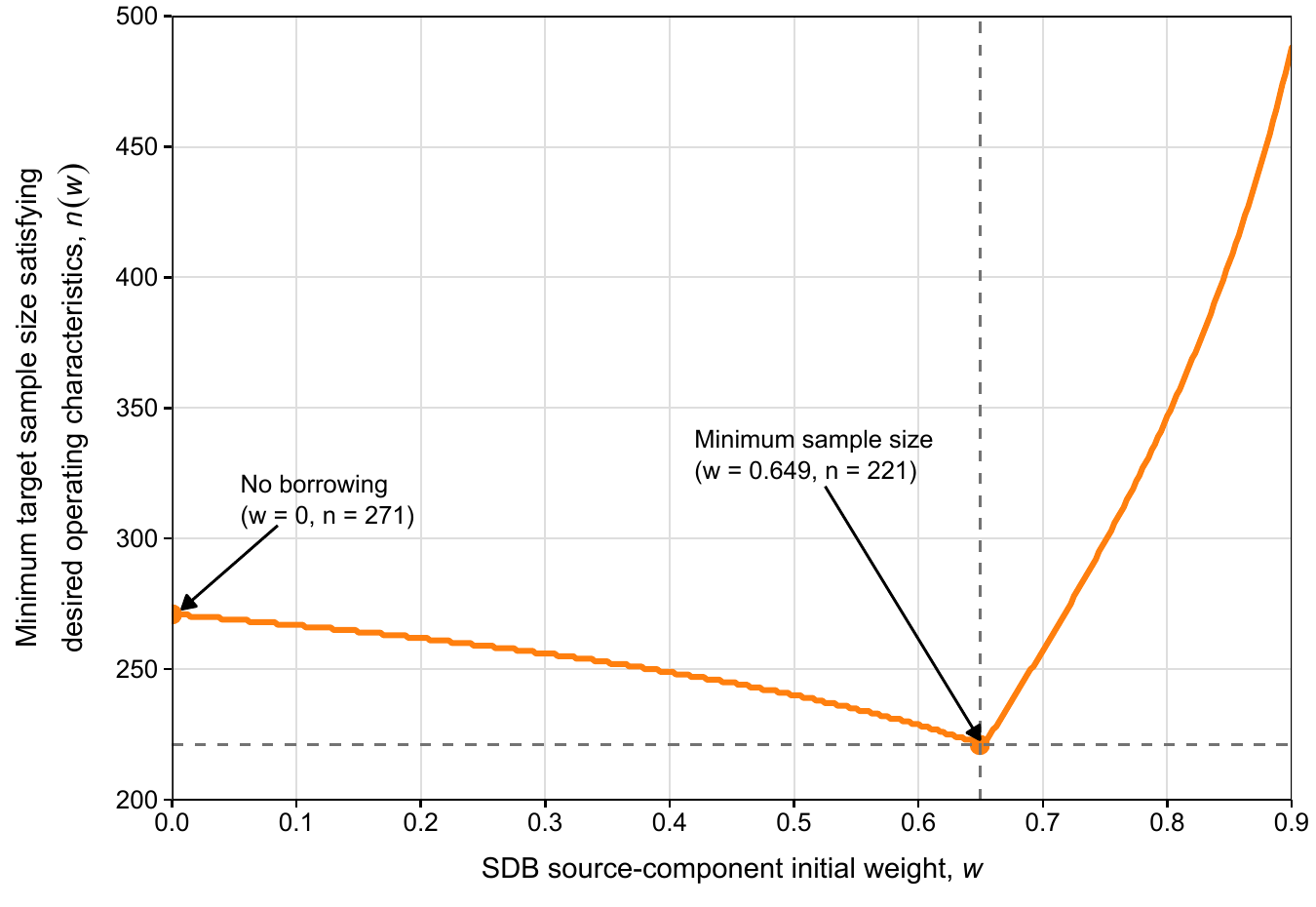}
\caption{Minimum target sample size $n(w)$ required to achieve at least 90\% power while controlling one-sided Type I error at 5\%, as a function of the SDB source-component initial weight $w$. The no-borrowing benchmark is shown at $w=0$, representing the limiting case $w\downarrow0$. }
\label{fig:sdb-nw}
\end{figure}

On the target-information scale, the source-prior ELIR ESS is ${12.4}/{0.017}=729.412$. The selected cap $w_{\mathrm{des}}=0.649$ therefore limits the fraction of the source prior's ELIR ESS that can be borrowed to $0.649$, corresponding to at most $0.649(729.412)\approx 474$ adolescent information units. These quantities are expressed on the adolescent information scale and should not be interpreted as fractions or numbers of adult participants.

In practice, the requirements of power and Type I error may depend on disease rarity, unmet medical need, and recruitment feasibility, and should be discussed and agreed upon between the trial sponsor and the regulatory agency.  Figure~\ref{fig:sdb-design-map} maps the design sample size and the corresponding design weight over a broad range of power and Type I error requirements. For each pair of requirements, $n_{\mathrm{des}}$ is the minimum integer sample size satisfying both criteria, and $w_{\mathrm{des}}$ is a corresponding design weight specifying the maximum allowable borrowing fraction of source information.

\begin{figure}[!htbp]
\centering
\includegraphics[width=1\linewidth]{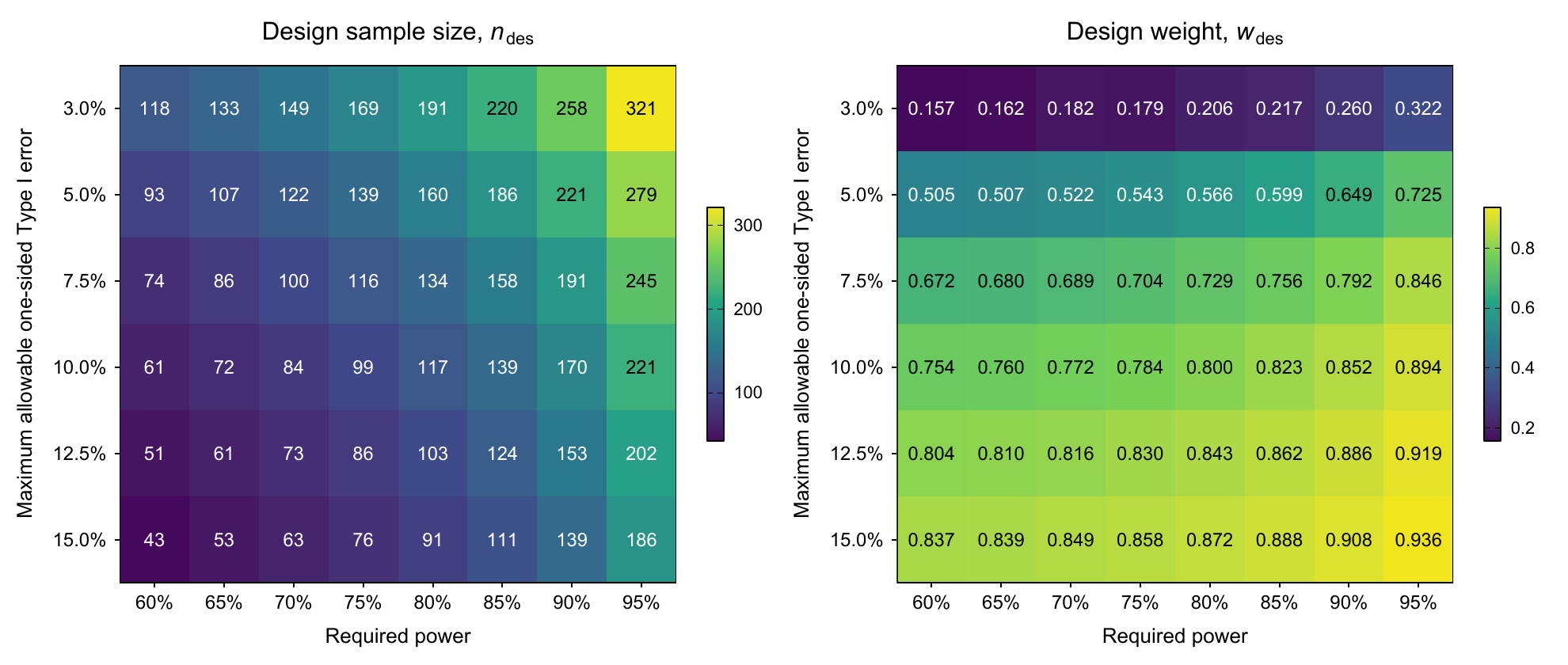}
\caption{Joint SDB feasibility map. Each cell in the left panel shows the minimum sample size satisfying the stated power and Type I error requirements. Each cell in the right panel shows the smallest SDB initial weight attaining the corresponding sample size.}
\label{fig:sdb-design-map}
\end{figure}

As expected, higher required power or a smaller allowable Type I error generally increases $n_{\mathrm{des}}$. The corresponding $w_{\mathrm{des}}$, representing the maximum allowable source-component posterior weight, generally increases with a higher required power or a larger allowable Type I error. 
The displayed scenarios illustrate the tradeoff among recruitment burden, operating characteristics, and allowable posterior borrowing and show how SDB design under prespecified operating-characteristic constraints can support feasibility discussions.

\subsection{Borrowing and Decision Profiles: SDB Versus RMP}
\label{sec:best-borrowing-decision}

We can also calibrate RMPs using the same power and Type I error criteria as for SDB. Consider three RMPs with weakly informative priors containing 1, 0.1, and 0.0001 units of target information, respectively, i.e.,
\[
\theta\mid M_w
\sim
\Normal\left(0,\frac{12.4}{k}\right),
\qquad
k\in\{1,0.1,0.0001\}.
\]
All three RMP specifications attain the same minimum $n=221$, with calibrated source-component prior weights of $0.1252$, $0.0431$, and $0.0014$, respectively. Thus, although a carefully calibrated
RMP can achieve operating characteristics similar to those of SDB, the required RMP prior weights differ substantially and depend strongly on the specification of weakly informative prior.

Figure~\ref{fig:weights-overlay} compares the borrowing profiles between SDB and RMP. In the left panel, SDB is calibrated directly to achieve at least 90\% power while controlling one-sided Type I error at 5\%, yielding a source-component initial weight of approximately 0.649; for comparison, each RMP is assigned a source-component prior weight with the same numerical value, 0.649, without further calibration. SDB caps the source-component posterior weight at 0.649, whereas the RMPs assigned this same numerical value for their prior mixture weights can continue to assign very large posterior weights under substantial source--target conflict. Thus, the numerical value of an RMP source-component prior weight does not by itself constrain realized posterior borrowing, and the resulting decision boundary depends strongly on the weakly-informative-prior specification. In the right panel, SDB and each RMP are calibrated separately using the same operating-characteristic criteria, yielding borrowing and decision profiles that closely agree, although the required RMP source-component prior weights vary substantially with the weakly-informative-prior specification. 

\begin{figure}[!htbp]
\centering
\includegraphics[width=1\linewidth]{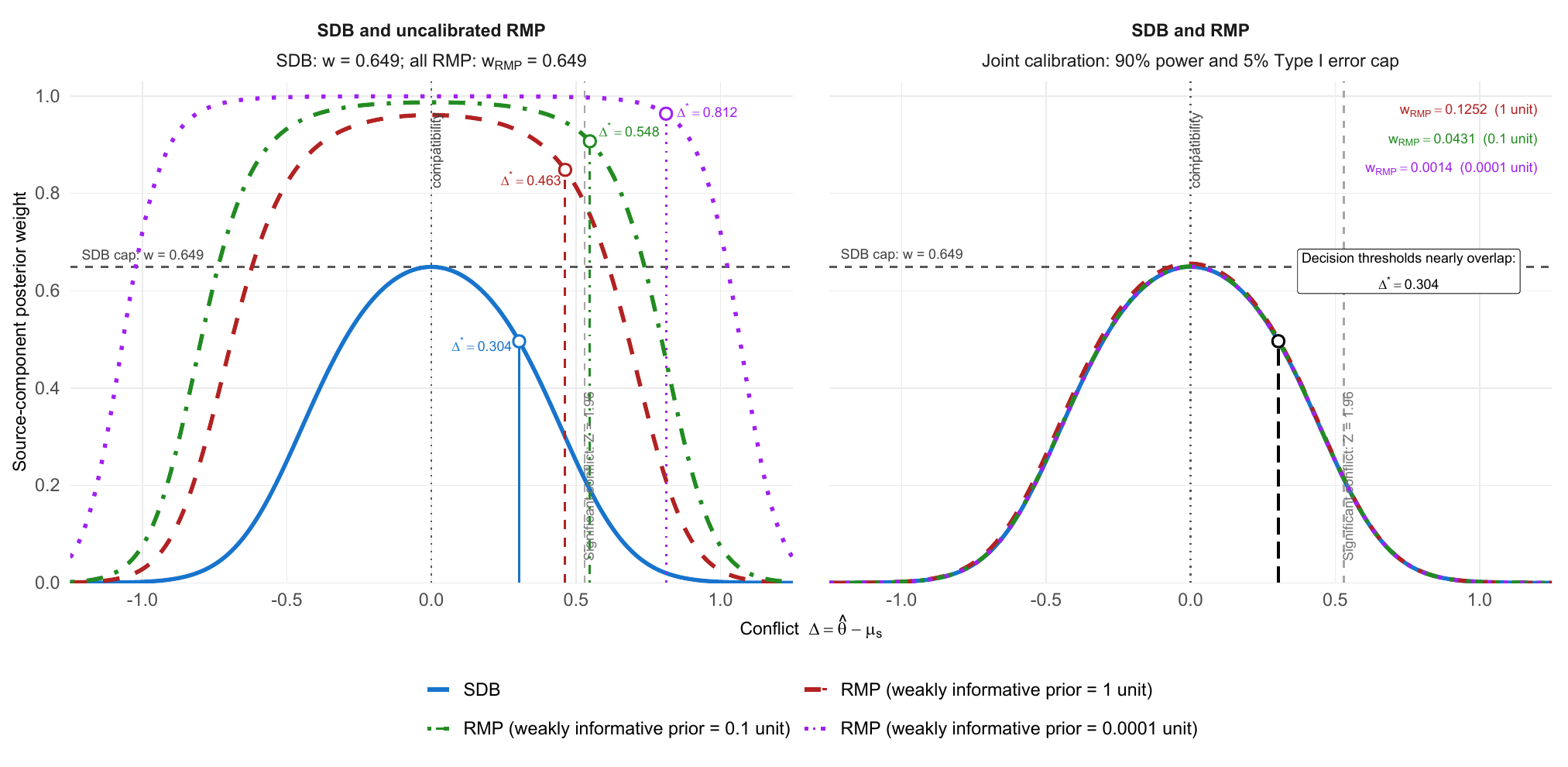}
\caption{Source-component posterior weight versus conflict
$\Delta=\hat\theta-\mu_s$ at $n=221$. Left: SDB is calibrated directly
to achieve at least 90\% power while controlling one-sided Type I error
at 5\%, yielding a source-component initial weight of 0.649; each RMP is
then assigned a source-component prior weight with the same numerical value, 0.649, without further calibration. Method-specific colored vertical lines indicate the
resulting efficacy boundaries. Right: SDB and each RMP are calibrated
separately using the same operating-characteristic criteria. The black
vertical line indicates their nearly identical efficacy boundaries.
In both panels, the rejection region lies to the left of the
corresponding efficacy boundary. Gray vertical lines indicate exact
compatibility and the attenuation-side conflict $Z=1.96$.}
\label{fig:weights-overlay}
\end{figure}
Suppose the planned trial produced the adolescent estimate reported by Best et al.,
\[
\hat\theta=-0.395,
\qquad \text{i.e.,} \quad
\Delta=\hat\theta-\mu_s=0.299.
\]
As a no-borrowing benchmark, the corresponding conventional one-sided frequentist $z$-test has a Type I error level of 0.025. At $n=221$ and $\hat\theta=-0.395$, the test statistic is $z=-1.67$, which does not cross the efficacy threshold $-1.96$. In contrast, SDB and the three calibrated RMPs assign source-component posterior weights of approximately $0.50$ and yield nearly identical results, with $\mathbb P(\theta<0\mid y)=0.9762$ and a 97.5\% posterior upper limit of $-0.0056$, thereby meeting the efficacy criterion. This conclusion is also evident from their nearly identical rejection regions on the conflict scale: the observed $\Delta=0.299$ lies just within the efficacy boundary of approximately $\Delta=0.304$.

\subsection{Tipping Point Analysis}

Once target data are available, it is useful to conduct a tipping-point analysis to assess the sensitivity of the efficacy conclusion to the strength of borrowing. For SDB and RMP, while the corresponding initial or prior weight is specified prospectively, the realized borrowing is reflected in the source-component posterior weight, which is dynamically determined by the observed target data. We therefore define the tipping point on the source-component posterior-weight scale. 
Specifically, the tipping point is the minimum source-component posterior weight required for the observed target data to satisfy the prespecified posterior decision criterion. The corresponding SDB initial weight or RMP prior weight can then be determined conditional on the observed data.

For $n=221$ and $\hat\theta=-0.395$, Figure~\ref{fig:tipping_point} displays the posterior median and 95\%
credible interval as functions of the source-component posterior weight for SDB and for the RMP with a 1-unit weakly informative prior. The SDB decision tips at a source-component posterior weight of approximately 0.476. For the 1-unit RMP, the corresponding tipping weight is approximately 0.480. Their similarity reflects the fact that, for this target dataset, the RMP weakly-informative-component posterior is close to the SDB no-borrowing posterior. 

The SDB initial weight and RMP prior weight corresponding to their respective tipping points are method-specific and have different interpretations. For SDB, the minimum initial weight required for success is approximately $w=0.626$, which has a direct posterior-scale interpretation because $w$ is the maximum source-component posterior weight attained at compatibility.  For the 1-unit RMP, the corresponding minimum source-component prior weight is approximately $w_{\mathrm{RMP}}=0.115$, representing the prior probability assigned to the source-informed model relative to the weakly informative model; however, its numerical value does not by itself directly indicate the corresponding posterior weight, which also depends on the specification of the weakly informative prior. 

\begin{figure}[H]
\centering
\includegraphics[width=1\linewidth]{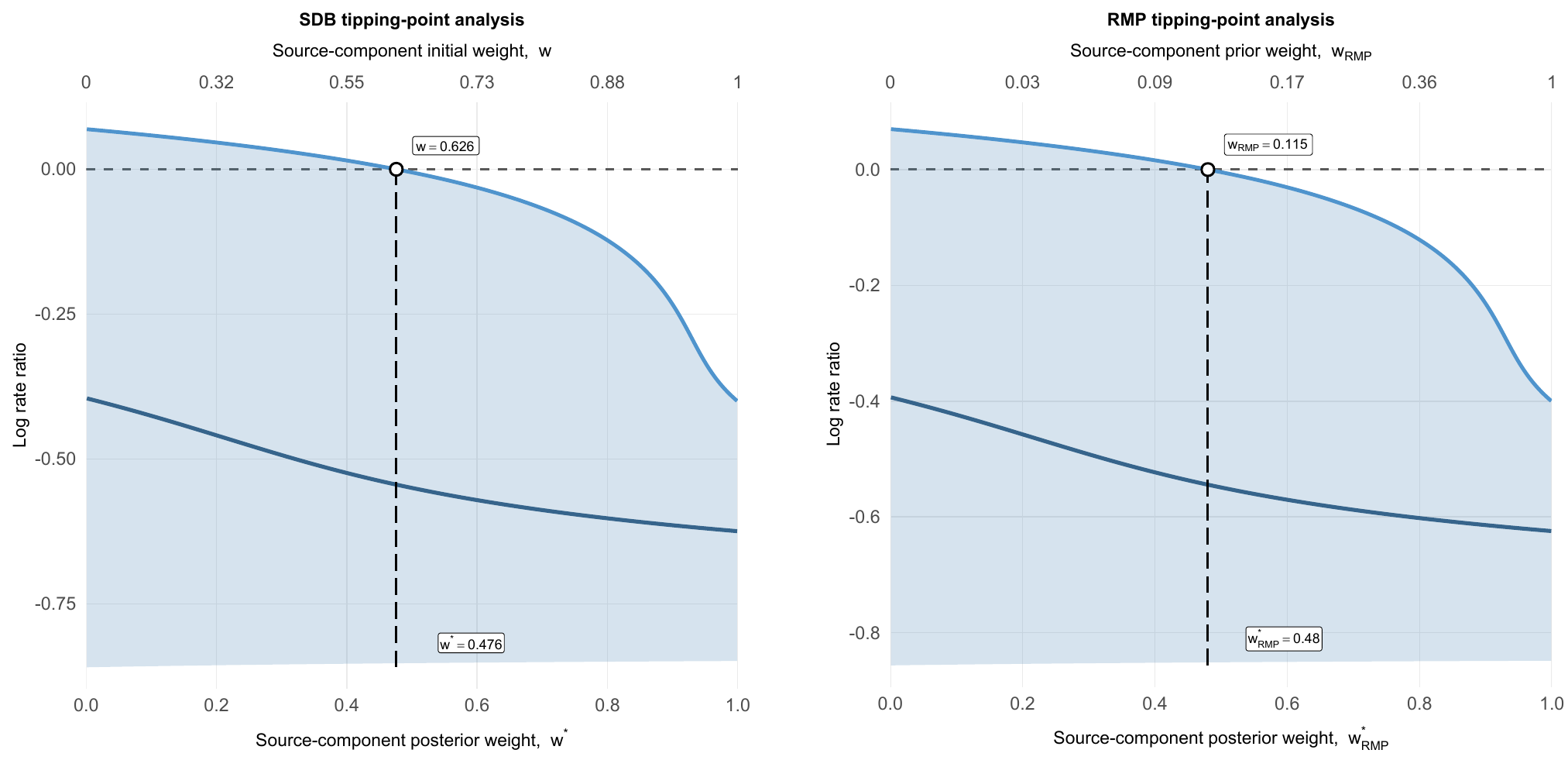}
\caption{Tipping-point analyses for SDB (left) and the RMP with a
1-unit weakly informative prior (right), with $n=221$ and
$\hat{\theta}=-0.395$. In each panel, the shaded ribbon represents the
95\% posterior credible interval, the dark blue line represents the posterior median, and the upper blue line represents the 97.5\% posterior upper limit. The vertical dashed line and open circle identify the tipping point at which the 97.5\% posterior upper limit reaches zero. The corresponding source-component posterior tipping weights are annotated near the bottom of each vertical line; the SDB initial tipping weight (left) and the RMP prior tipping weight (right) are annotated near the top.}
\label{fig:tipping_point}
\end{figure}

\section{Binomial Numerical Illustration}
\label{sec:binomial-illustration}

To illustrate the joint SDB design for a discrete endpoint, consider a hypothetical single-arm pediatric target study with a favorable response probability $p$. Suppose that an adult source study observes $x_s=36$
responses among $n_s=120$ subjects. Under the normalized-source-likelihood construction,
\[
p\mid M_s
\sim
\mathrm{Beta}(36,84).
\]

For the pediatric study, let
\[
X\mid p
\sim
\mathrm{Binomial}(n,p),
\]
and consider
\[
H_0:p\leq p_0,
\qquad
H_1:p>p_0,
\qquad
p_0=0.15.
\]
The source-compatible planning alternative is $p_A=0.30$. Set $\alpha_0=0.025$ and declare efficacy under SDB when
\[
\mathbb P_{\mathrm{SDB}}(p>p_0\mid X)
\geq
1-\alpha_0
=
0.975.
\]
As shown in Section~\ref{sec:binomial-case}, under no borrowing the SDB decision rule induced by this posterior threshold has an asymptotic frequentist Type I error of $0.025$.

For comparison with RMP, consider two proper weakly informative components for binary endpoints:
\[
p\mid M_w
\sim
\mathrm{Beta}(1,1)
\qquad\text{or}\qquad
p\mid M_w
\sim
\mathrm{Beta}\left(\frac{1}{2},\frac{1}{2}\right),
\]
corresponding to the uniform and Jeffreys priors, respectively~\citep{schmidli2014robust}. Each RMP also uses the same posterior criterion
\[
\mathbb P_{\mathrm{RMP}}(p>p_0\mid X)
\geq
0.975.
\]
For each method, the target sample size and corresponding source-component weight---the initial weight for SDB and the prior weight for RMP---are jointly selected to achieve at least 90\% power at $p_A=0.30$ while controlling the one-sided Type~I error at $p_0$ to no more than $\alpha_{\max}=0.05$. The calibration is performed separately for SDB and for each RMP specification. 
As a conventional benchmark, we use the exact one-sided binomial test at nominal level $0.025$. Its sample size is selected to achieve at least 90\% power at $p_A=0.30$.
For SDB, boundary counts are handled according to the convention in Section~\ref{sec:binomial-case}. Power and Type~I error are evaluated exactly under the binomial sampling distribution. The results are shown in Table~\ref{tab:binomial-design}. For each borrowing method, the smallest weight among those attaining the minimum sample size is selected.

\begin{table}[!htbp]
\centering
\scriptsize
\setlength{\tabcolsep}{3.5pt}
\caption{Conventional exact binomial benchmark and jointly calibrated
borrowing designs in the binomial numerical illustration. Power is
evaluated at $p_A=0.30$, and Type~I error is evaluated at $p_0=0.15$.}
\label{tab:binomial-design}
\begin{tabular}{lcccccc}
\toprule
Method
&
\shortstack{Target\\sample size}
&
\shortstack{Selected\\weight}
&
\shortstack{Maximum posterior\\source-component weight}
&
\shortstack{Efficacy\\rule}
&
Power
&
\shortstack{Type~I\\error}
\\
\midrule
Exact binomial test
&
79
&
--
&
--
&
$X\geq19$
&
0.901
&
0.023
\\
SDB
&
64
&
0.550
&
0.550
&
$X\geq15$
&
0.903
&
0.049
\\
RMP, uniform
&
64
&
0.033
&
0.164
&
$X\geq15$
&
0.903
&
0.049
\\
RMP, Jeffreys
&
64
&
0.072
&
0.386
&
$X\geq15$
&
0.903
&
0.049
\\
\bottomrule
\end{tabular}

\vspace{0.4em}
\begin{minipage}{0.96\linewidth}
\footnotesize
\textit{Note: For SDB, the selected weight is the source-component initial weight $w$; for RMP, it is the source-component prior weight $w_{\mathrm{RMP}}$. The maximum posterior source-component weight is
evaluated over the interior target counts at the selected target sample size and, for all three borrowing methods in this example, occurs at $x=19$. Within each borrowing method, the smallest weight attaining the minimum
sample size is reported.} 
\end{minipage}
\end{table}

For the conventional exact one-sided binomial test at nominal level $0.025$, the smallest sample size achieving at least 90\% power at $p_A=0.30$ is $n=79$, with efficacy declared for $X\geq19$. The corresponding Type~I error is $0.023$, slightly below $0.025$ because of the discreteness of the binomial outcome. 

The jointly selected SDB design has $n_{\mathrm{des}}=64$ and $w_{\mathrm{des}}=0.550$, with efficacy declared for $X\geq15$. The separately calibrated RMP designs with the uniform and Jeffreys weakly informative components both select the same target sample size of 64 and the same efficacy rule $X\geq15$, and therefore have the same power and Type I error in this example. 

For $n_{\mathrm{des}}=64$, the compatibility-anchor formula in Section~\ref{sec:binomial-case} gives
\[
x_c
=
\left\lceil
(64-1)\frac{36}{120}
\right\rceil
=
19,
\]
corresponding to a target response rate of $19/64=0.297$, numerically close to the source response rate $36/120=0.3$. At this count, the SDB source-component posterior weight attains its selected cap of
$0.550$. In this example, the posterior source-component weights for the two RMPs are also maximized at $x=19$. Their selected prior weights of $0.033$ and $0.072$ yield maximum posterior source-component weights
of $0.164$ and $0.386$ for the uniform and Jeffreys specifications, respectively. Thus, unlike the SDB initial weight, an RMP prior weight is not itself the maximum posterior source-component weight; the latter
depends on both the prior weight and the selected weakly informative prior.

\section{Discussion}

SDB can be viewed as a standardized modification of the RMP framework. By combining the source-informed and no-borrowing components through a Bayes factor standardized by its supremum over the admissible data space, SDB eliminates the need to specify a separate weakly informative prior and makes the source-component initial weight a sharp uniform upper bound on the source-component posterior weight.

This construction provides several useful interpretations of borrowing. Under the conditions of Theorem~\ref{thm:sdb-elir}, the posterior ELIR
decomposition identifies the borrowed source-prior ESS and gives the
source-component initial weight a direct interpretation as the maximum
allowable fraction of the source-prior ELIR ESS that can be borrowed
in the target analysis.
In the normal setting, SDB additionally provides a simple geometric interpretation of dynamic borrowing:  the posterior odds of the source component are attenuated according to a standard normal kernel in the conflict $Z$-score. The conflict $Z$-score scales the observed conflict against the total uncertainty---comprising both target sampling error and source-prior variance---anticipated under the source model. Consequently, borrowing is maximal at the compatibility anchor, the source-component posterior odds retain approximately 61\% of their initial value at $|Z|=1$ and approximately 15\% at $|Z|=1.96$, with stronger attenuation for larger conflicts. 

Covariate adjustment can be incorporated in SDB by specifying the target likelihood $L(\theta; y)$ from a covariate-adjusted model and using a source prior $\pi(\theta|M_s)$ for the same covariate-adjusted estimand. For a normal likelihood and normal source prior, the normal-case SDB results apply directly. For non-normal cases, this can be done by using normal approximations to the covariate-adjusted likelihood and source prior.

From a design perspective, the source-component initial weight and the sample size for SDB can be jointly selected under simultaneous prespecified constraints on power and Type I error. The same operating-characteristic calibration can be applied to RMP. In the normal case study, the SDB and all three RMPs calibrated under the same operating-characteristic constraints achieved similar borrowing and decision profiles. In the binomial illustration, SDB and the two RMP specifications also selected the same target sample size and efficacy rule after separate calibration. Thus, the principal advantages of SDB over RMP under careful calibration are not superior attainable operating characteristics, but rather the elimination of a separate weakly informative prior specification, a borrowing parameter with a direct posterior-cap interpretation, sharp uniform control of the source-component posterior weight, an exact no-borrowing limit as the source-component posterior weight approaches zero, a simple ELIR decomposition without an additional weakly-informative-prior contribution, a straightforward illustration of tradeoffs among operating characteristics and maximum allowable posterior borrowing, and a geometric interpretation of borrowing attenuation in the normal setting.

Several limitations should be recognized. First, the standardized construction requires the integrated quantities defining the source-informed and no-borrowing components to be finite and positive over the admissible data space, together with the stated conditions on the source prior; boundary observations may require a prespecified decision convention, as in the binomial example.
Second, the ELIR decomposition and monotone-power theorem rely on the stated one-parameter canonical exponential family conditions. Applications outside these settings may require numerical evaluation and do not automatically inherit the same theoretical properties. Third, SDB controls borrowing only after the source and target populations have been judged sufficiently comparable and the estimand has been appropriately aligned. The standardized compatibility measure and borrowing cap do not replace the clinical assessment of transportability. 

Taken together, SDB provides a practical framework for using credible source information when target data are limited. Its intuitive compatibility measure and sharp posterior-weight cap make dynamic borrowing more transparent to prespecify, evaluate, and communicate, while preserving data-adaptive borrowing according to source--target compatibility.

\section{Conclusion}

SDB provides a direct and interpretable formulation of dynamic borrowing. Through standardization of the source-informed versus no-borrowing comparison, it makes the source-component initial weight a sharp uniform upper bound on the source-component posterior weight and avoids specification of a separate weakly informative prior. Under the stated one-parameter canonical exponential-family conditions, the posterior ELIR decomposition provides a target-information-scale interpretation of this borrowing cap; in the normal model, the source-component posterior odds are attenuated according to a standard normal kernel in the conflict $Z$-score. Together with joint calibration of the borrowing cap and target sample size under prespecified operating-characteristic constraints, these properties provide a transparent framework for prespecifying, evaluating, and communicating dynamic borrowing. 

\bmsection*{Acknowledgments}
The authors express their gratitude to Dr. James Travis and Dr. Frank Harrell for their valuable and constructive feedback during the development of this work.

\bmsection*{Disclaimer}
The contents, views or opinions expressed in this publication or presentation are those of the authors and do not necessarily reflect official policy or position of the U.S. Food and Drug Administration. Mention of trade names, commercial products, or organizations does not imply endorsement by the U.S. Government.

\bmsection*{Funding}

No specific funding was received for this work.

\bmsection*{Disclosure statement}

The authors declare no potential conflict of interest.

\bibliography{bibliography}



\setcounter{theorem}{0}
\setcounter{equation}{0} 

\newpage
\appendix

\begin{center}
{\Large\bfseries Appendix}
\end{center}

\section{SDB Weight Interpretation for Improper Reference Priors}

\textbf{Claim:} When the flat reference prior is improper, for each $y\in\mathcal Y_{\mathrm{adm}}$, the SDB source-component posterior weight $w^*(y)$ is the limit of posterior source-model probabilities obtained under a sequence of proper reference priors with calibrated prior model weights.

\begin{proof}
Let
\[
O_{\max}
=
\sup_{u\in\mathcal Y_{\mathrm{adm}}}O(u),
\]
which is finite and strictly positive under the conditions in Section~\ref{sec:sdb-method}. Let $K_m\uparrow\Theta$ be an increasing sequence of compact sets with $0<|K_m|<\infty$, and define the proper reference priors
\[
\pi_{K_m}(\theta)
=
\frac{\mathbf 1_{K_m}(\theta)}{|K_m|}.
\]

Fix $y\in\mathcal Y_{\mathrm{adm}}$. The corresponding reference-model marginal likelihood is
\[
f_{K_m}(y\mid M_{nb})
=
\int_\Theta L(\theta;y)\pi_{K_m}(\theta)\,d\theta
=
\frac{1}{|K_m|}
\int_\Theta
\mathbf 1_{K_m}(\theta)L(\theta;y)\,d\theta.
\]
Since $K_m\uparrow\Theta$ and $I(y)=\int_\Theta L(\theta;y)\,d\theta>0$, the integral is strictly positive for all sufficiently large $m$; we consider such $m$ below. Hence, the source-to-reference Bayes factor is
\[
\mathrm{BF}_{K_m}(y)
=
\frac{f(y\mid M_s)}
     {f_{K_m}(y\mid M_{nb})}
=
\frac{|K_m|f(y\mid M_s)}
{\int_\Theta
\mathbf 1_{K_m}(\theta)L(\theta;y)\,d\theta}.
\]

For each $m$, define the calibrated prior model odds by
\[
\frac{\tilde w_m}{1-\tilde w_m}
=
\frac{w}{1-w}
\frac{1}{|K_m|O_{\max}}.
\]
Because $w\in(0,1)$, $0<O_{\max}<\infty$, and $0<|K_m|<\infty$, this defines a proper prior model weight
$\tilde w_m\in(0,1)$.

The posterior source-model odds under the proper reference prior $\pi_{K_m}$ are therefore
\begin{align*}
\frac{
\mathbb P(M_s\mid y,K_m,\tilde w_m)
}{
1-\mathbb P(M_s\mid y,K_m,\tilde w_m)
}
&=
\frac{\tilde w_m}{1-\tilde w_m}
\mathrm{BF}_{K_m}(y)
=
\frac{w}{1-w}
\frac{f(y\mid M_s)}
{O_{\max}
\int_\Theta
\mathbf 1_{K_m}(\theta)L(\theta;y)\,d\theta}.
\end{align*}

Since $K_m\uparrow\Theta$ and the likelihood is nonnegative, the Monotone Convergence Theorem gives
\[
\int_\Theta
\mathbf 1_{K_m}(\theta)L(\theta;y)\,d\theta
\uparrow
\int_\Theta L(\theta;y)\,d\theta
=
I(y).
\]
Therefore,
\begin{align*}
\frac{
\mathbb P(M_s\mid y,K_m,\tilde w_m)
}{
1-\mathbb P(M_s\mid y,K_m,\tilde w_m)
}
\longrightarrow
\frac{w}{1-w}
\frac{f(y\mid M_s)}{O_{\max}I(y)}
=
\frac{w}{1-w}
\frac{O(y)}{O_{\max}}
=
\frac{w}{1-w}\mathrm{SBF}(y)=\frac{w^*(y)}{1-w^*(y)}.
\end{align*}
Hence, by continuity of the map from odds to probability,
\[
\mathbb P(M_s\mid y,K_m,\tilde w_m)
\longrightarrow
w^*(y).
\]
Thus, $w^*(y)$ is the limiting posterior source-model probability under a sequence of proper reference priors with calibrated prior model weights.
\end{proof}

\section{Proof of Theorem~\ref{thm:sdb-elir}}

\begin{proof}
Under the stated canonical exponential-family assumptions,
\[
\frac{d^2}{d\theta^2}\log L_s(\theta;y_s)
=
-\frac{n_s}{\phi_s}\psi''(\theta)
\leq0,
\]
so the source likelihood is log-concave in $\theta$. Since it is
integrable by assumption, its normalized density is proper and bounded.
Thus, the source prior satisfies the boundedness condition required for
the SDB construction.

Fix $y$ and write $w^*(y)$ as $w^*$,
$\pi(\theta\mid M_s,y)$ as $\pi_{s,y}(\theta)$,
$\pi(\theta\mid M_{nb},y)$ as $\pi_{nb,y}(\theta)$, and
$\pi_{\mathrm{SDB}}(\theta\mid y)$ as $\pi_{\mathrm{SDB},y}(\theta)$. Then
\[
\pi_{\mathrm{SDB},y}
=
w^*\pi_{s,y}+(1-w^*)\pi_{nb,y}.
\]
Define the component score functions and curvatures by
\begin{align*}
s_s(\theta;y)&=\frac{d}{d\theta}\log\pi_{s,y}(\theta),
&
s_{nb}(\theta;y)&=\frac{d}{d\theta}\log\pi_{nb,y}(\theta),\\
\kappa_s(\theta;y)&=-\frac{d^2}{d\theta^2}\log\pi_{s,y}(\theta),
&
\kappa_{nb}(\theta;y)&=-\frac{d^2}{d\theta^2}\log\pi_{nb,y}(\theta).
\end{align*}
The responsibility functions are
\[
r_s(\theta;y)=\frac{w^*\pi_{s,y}(\theta)}{\pi_{\mathrm{SDB},y}(\theta)},
\qquad
r_{nb}(\theta;y)=1-r_s(\theta;y).
\]

Using the information formula for mixture distributions~\citep{neuenschwander2020predictively}, the curvature of the posterior mixture is:
\begin{align*}
i_{\mathrm{SDB}}(\theta\mid y)
= &\frac{1}{\pi_{\mathrm{SDB},y}^2}
\left(w^* \pi_{s,y} s_s + (1 - w^*)\pi_{nb,y} s_{nb}\right)^2 \\
&- \frac{1}{\pi_{\mathrm{SDB},y}}
\left(w^* \pi_{s,y}(s_s^2 - \kappa_s)
+ (1 - w^*)\pi_{nb,y}(s_{nb}^2 - \kappa_{nb})\right) \\
=& (r_s s_s + r_{nb} s_{nb})^2
- (r_s s_s^2 + r_{nb} s_{nb}^2) 
 + r_s\kappa_s + r_{nb}\kappa_{nb} \\
=& r_s\kappa_s + r_{nb}\kappa_{nb}
- r_s r_{nb} (s_s - s_{nb})^2. 
\quad \text{(using $r_s + r_{nb} = 1$)}
\end{align*}

Therefore,
\begin{align*}
\mathrm{ELIR}\{\pi_{\mathrm{SDB}}(\cdot\mid y)\}
&=
\mathbb E_{\theta\sim\pi_{\mathrm{SDB},y}}
\left[\frac{i_{\mathrm{SDB}}(\theta\mid y)}{i_1(\theta)}\right]\\
&= \mathbb E_{\theta\sim\pi_{\mathrm{SDB},y}}\left[\frac{r_s\kappa_s}{i_1(\theta)}\right] + \mathbb E_{\theta\sim\pi_{\mathrm{SDB},y}}\left[\frac{r_{nb}\kappa_{nb}}{i_1(\theta)}\right] - \mathbb E_{\theta\sim\pi_{\mathrm{SDB},y}}\left[\frac{r_s r_{nb} (s_s - s_{nb})^2}{i_1(\theta)}\right] \\
&=w^*\mathbb E_{\theta\sim\pi_{s,y}}
\left[\frac{\kappa_s(\theta;y)}{i_1(\theta)}\right]
+(1-w^*)\mathbb E_{\theta\sim\pi_{nb,y}}
\left[\frac{\kappa_{nb}(\theta;y)}{i_1(\theta)}\right]
-\gamma(y)\\
&=
w^*\mathrm{ELIR}\{\pi_{s,y}\}
+
(1-w^*)\mathrm{ELIR}\{\pi_{nb,y}\}
-
\gamma(y),
\end{align*}
where
\begin{align*}
\gamma(y)
&=
\mathbb E_{\theta\sim\pi_{\mathrm{SDB},y}}
\left[
\frac{r_s r_{nb}}{i_1}(s_s-s_{nb})^2
\right]
.
\end{align*}
Because $r_s,r_{nb}\geq0$ and $i_1(\theta)>0$ under the regular canonical exponential-family assumptions, $\gamma(y)\geq0$.

Under the stated canonical exponential-family assumptions,
\[
\pi_{s,y}(\theta)\propto L_s(\theta;y_s)L(\theta;y),
\qquad
\pi_{nb,y}(\theta)\propto L(\theta;y),
\]
and hence
\begin{align*}
\kappa_s(\theta; y) &= -\frac{\partial^2}{\partial \theta^2} \log \pi_{s,y} = -\frac{\partial^2}{\partial \theta^2} \left(\log L_s(\theta; y_s) + \log L(\theta; y)\right)= \left(\frac{n_s}{\phi_s} + \frac{n}{\phi}\right)\psi''(\theta), \\
\kappa_{nb}(\theta; y) &= -\frac{\partial^2}{\partial \theta^2} \log \pi_{nb,y} = -\frac{\partial^2}{\partial \theta^2} \log L(\theta; y) = \frac{n}{\phi}\psi''(\theta).
\end{align*}
Since the source-informed prior is $\pi(\theta\mid M_s)\propto L_s(\theta;y_s)$, its log–density differs from $\log L_s(\theta;y_s)$ only by a constant. Hence its curvature is
$$
-\frac{d^2}{d\theta^2}\log \pi(\theta\mid M_s) = \frac{n_s}{\phi_s}\psi''(\theta),
$$
and, by the ELIR definition relative to the target model,
$$
\mathrm{ELIR}\{\pi(\cdot\mid M_s)\}
= \mathbb{E}_{\theta\sim\pi(\cdot\mid M_s)}\left[\frac{(n_s/\phi_s)\psi''(\theta)}{i_1(\theta)}\right]
= \mathbb{E}_{\theta\sim\pi(\cdot\mid M_s)}\left[\frac{(n_s/\phi_s)\psi''(\theta)}{\psi''(\theta)/\phi}\right]
= n_s\frac{\phi}{\phi_s}.
$$

Therefore, the ELIR for the source-component posterior is:
\begin{align*}
\mathrm{ELIR}\{\pi_{s,y}\}
&= \mathbb{E}_{\theta \sim \pi_{s,y}}
\left[\frac{\kappa_s(\theta; y)}{i_1(\theta)}\right] \\
&= \mathbb{E}_{\theta \sim \pi_{s,y}}
\left[\frac{(n_s/\phi_s + n/\phi)\psi''(\theta)}
{\psi''(\theta)/\phi}\right] \\
&= n_s\frac{\phi}{\phi_s} + n \\
&= \mathrm{ELIR}\{\pi(\cdot\mid M_s)\} + n.
\end{align*}
Similarly, for the no-borrowing component posterior:
\begin{align*}
\mathrm{ELIR}\{\pi_{nb,y}\} &= \mathbb{E}_{\theta \sim \pi_{nb,y}}\left[\frac{\kappa_{nb}(\theta; y)}{i_1(\theta)}\right] = \mathbb{E}_{\theta \sim \pi_{nb,y}}\left[\frac{(n/\phi)\psi''(\theta)}{\psi''(\theta)/\phi}\right] = n.
\end{align*}

Substituting these results into the general expression yields
\begin{align*}
\mathrm{ELIR}\{\pi_{\mathrm{SDB}}(\cdot\mid y)\} &= w^*\left(\mathrm{ELIR}\{\pi(\cdot\mid M_s)\}+ n\right) + (1 - w^*)n - \gamma(y) \\
&= n + w^*\mathrm{ELIR}\{\pi(\cdot\mid M_s)\} - \gamma(y).
\end{align*} 
\end{proof}

\section{Proof of Theorem~\ref{thm:monotone-power}}

\begin{proof}
Fix $n$ and $w$. By the definition of the standardized Bayes factor,
\begin{align*}
\frac{w^*(y)}{1-w^*(y)}
=
\frac{w}{1-w}\mathrm{SBF}(y)
=
\frac{w}{1-w}
\frac{f(y\mid M_s)}{O_{\max}I(y)}
=
\lambda_{n,w}\frac{f(y\mid M_s)}{I(y)},
\end{align*}
where
\[
\lambda_{n,w}
=
\frac{w}{1-w}\frac{1}{O_{\max}}>0.
\]
For fixed $n$, $O_{\max}$ is independent of the observed data $y$.
Thus, for fixed $n$ and $w$, $\lambda_{n,w}$ does not depend on $y$.

Substituting the component posteriors into the SDB mixture posterior gives
\begin{align*}
\pi_{\mathrm{SDB}}(\theta\mid y)
&=
w^*(y)
\frac{L(\theta;y)\pi(\theta\mid M_s)}
{f(y\mid M_s)}
+
(1-w^*(y))
\frac{L(\theta;y)}{I(y)}
\\
&=
\frac{1-w^*(y)}{I(y)}
L(\theta;y)
\left(
1+
\frac{w^*(y)}{1-w^*(y)}
\frac{I(y)}{f(y\mid M_s)}
\pi(\theta\mid M_s)
\right)
\\
&=
\frac{1-w^*(y)}{I(y)}
L(\theta;y)
\left(
1+\lambda_{n,w}\pi(\theta\mid M_s)
\right).
\end{align*}
Define
\[
h_{n,w}(\theta)
=
1+\lambda_{n,w}\pi(\theta\mid M_s).
\]

It remains to identify the normalizing constant. From
\[
\frac{w^*(y)}{1-w^*(y)}
=
\lambda_{n,w}\frac{f(y\mid M_s)}{I(y)},
\]
cross-multiplication gives
\[
w^*(y)I(y)
=
(1-w^*(y))\lambda_{n,w}f(y\mid M_s).
\]
Therefore,
\begin{align*}
(1-w^*(y))
\left(
I(y)+\lambda_{n,w}f(y\mid M_s)
\right)
=
(1-w^*(y))I(y)
+
w^*(y)I(y)
=
I(y),
\end{align*}
and hence
\[
\frac{1-w^*(y)}{I(y)}
=
\frac{1}
{I(y)+\lambda_{n,w}f(y\mid M_s)}.
\]
On the other hand,
\begin{align*}
\int_\Theta L(t;y)h_{n,w}(t)\,dt
=
\int_\Theta L(t;y)
\left(
1+\lambda_{n,w}\pi(t\mid M_s)
\right)dt
=
I(y)+\lambda_{n,w}f(y\mid M_s).
\end{align*}
Combining the preceding two identities gives
\[
\pi_{\mathrm{SDB}}(\theta\mid y)
=
\frac{L(\theta;y)h_{n,w}(\theta)}
{\int_\Theta L(t;y)h_{n,w}(t)\,dt}.
\]
Finally, since $\lambda_{n,w}>0$ and
$\pi(\theta\mid M_s)\geq0$,
\[
h_{n,w}(\theta)
=
1+\lambda_{n,w}\pi(\theta\mid M_s)
\geq1>0.
\]
Thus, for fixed $n$ and $w$, $h_{n,w}$ is a strictly positive kernel
that does not depend on the observed data $y$.

Let $U(Y)$ denote the scalar sufficient statistic and let $u$ be its observed value. For $u_2>u_1$,
\[
\frac{
\pi_{\mathrm{SDB}}(\theta\mid u_2;w)
}{
\pi_{\mathrm{SDB}}(\theta\mid u_1;w)
}
\propto_\theta
\frac{L(\theta;u_2)}{L(\theta;u_1)}.
\]
For a one-parameter canonical exponential family,
\[
\frac{L(\theta;u_2)}{L(\theta;u_1)}
\propto_\theta
\exp\left\{\frac{(u_2-u_1)\theta}{\phi}\right\},
\]
which is increasing in $\theta$. Therefore, the SDB posterior is stochastically increasing in $u$, and
\[
u
\longmapsto
\mathbb P_{\mathrm{SDB}}(\theta<\theta_0\mid u)
\]
is non-increasing. Consequently, the SDB success region is a lower set in $u$. By assumption, the prespecified convention for outcomes outside $\mathcal Y_{\mathrm{adm}}$ preserves this monotonicity over the full sampling support.

The target-data model also has a monotone likelihood ratio in $U(Y)$. Hence, for $\theta_2^\dagger>\theta_1^\dagger$, the sampling
distribution of $U(Y)$ under $\theta_2^\dagger$ is stochastically larger than that under $\theta_1^\dagger$. Because the success region
is a lower set,
\[
\rho(n,w,\theta_2^\dagger)
\leq
\rho(n,w,\theta_1^\dagger).
\]
Thus, $\rho(n,w,\theta^\dagger)$ is non-increasing in $\theta^\dagger$. Since
\[
\Theta_0
=
\{\theta^\dagger\in\Theta:\theta^\dagger\geq\theta_0\},
\]
it follows that
\[
\alpha(n,w)
=
\sup_{\theta^\dagger\in\Theta_0}
\rho(n,w,\theta^\dagger)
=
\rho(n,w,\theta_0).
\]
\end{proof}

\section{Normal Distribution Case}

\subsection{Derivation of Component Posteriors}

The source-component posterior is proportional to the product of source prior and target likelihood:
\begin{align*}
\pi(\theta \mid M_{s},\hat{\theta}) \propto \exp\left\{-\frac{\tau_s}{2}(\theta - \mu_s)^2\right\} \exp\left\{-\frac{\tau}{2}(\theta - \hat{\theta})^2\right\}.
\end{align*}
By completing the square in the exponent: 
\begin{align*}
\tau_s(\theta - \mu_s)^2 + \tau(\theta - \hat{\theta})^2
&= (\tau_s + \tau)\theta^2
- 2(\tau_s\mu_s + \tau\hat{\theta})\theta 
+ (\tau_s\mu_s^2 + \tau\hat{\theta}^2) \\
&= (\tau_s + \tau)
\left(\theta - \frac{\tau_s\mu_s + \tau\hat{\theta}}
{\tau_s + \tau}\right)^2
+ \text{constant}.
\end{align*}
Thus, the source-component posterior is normal with:
\begin{align*}
\text{Precision: } \tau_s + \tau, \qquad
\text{Mean: } \frac{\tau_s\mu_s + \tau\hat{\theta}}{\tau_s + \tau} = \mu_s + \frac{\tau}{\tau_s + \tau}\Delta.
\end{align*}

The no-borrowing component posterior is also normal:
$$\pi(\theta \mid M_{nb},\hat{\theta})=\frac{
L(\theta;\hat{\theta})
}{
\int_{\mathbb R}L(t;\hat{\theta})\,dt
} \quad\Leftrightarrow \quad\theta \mid M_{nb},\hat{\theta}\sim \Normal(\hat{\theta},\tau^{-1}).$$

\subsection{Derivation of the Compatibility Anchor and Posterior Weight}

Under the no-borrowing model, the normalized target likelihood is
\[
\pi(\theta\mid M_{nb},\hat{\theta})
=
\sqrt{\frac{\tau}{2\pi}}
\exp\left\{
-\frac{\tau}{2}
(\theta-\hat{\theta})^2
\right\}.
\]
The source prior is
\[
\pi(\theta\mid M_s)
=
\sqrt{\frac{\tau_s}{2\pi}}
\exp\left\{
-\frac{\tau_s}{2}
(\theta-\mu_s)^2
\right\}.
\]
Therefore, the overlap function is
\begin{align*}
O(\hat{\theta})
=
\int_{\mathbb R}
\pi(\theta\mid M_{nb},\hat{\theta})
\pi(\theta\mid M_s)\,d\theta
=
\frac{\sqrt{\tau\tau_s}}{2\pi}
\int_{\mathbb R}
\exp\left\{
-\frac{1}{2}
\left[
\tau(\theta-\hat{\theta})^2
+
\tau_s(\theta-\mu_s)^2
\right]
\right\}d\theta.
\end{align*}

To evaluate the integral, define
\[
m(\hat{\theta})
=
\frac{\tau\hat{\theta}+\tau_s\mu_s}
{\tau+\tau_s}.
\]
Completing the square gives
\begin{align*}
\tau(\theta-\hat{\theta})^2
+
\tau_s(\theta-\mu_s)^2
&=
(\tau+\tau_s)
\left(
\theta-m(\hat{\theta})
\right)^2+
\frac{\tau\tau_s}{\tau+\tau_s}
(\hat{\theta}-\mu_s)^2.
\end{align*}
Substituting this identity into the overlap integral yields
\begin{align*}
O(\hat{\theta})
=
\frac{\sqrt{\tau\tau_s}}{2\pi}
\exp\left\{
-\frac{\tau\tau_s}{2(\tau+\tau_s)}
(\hat{\theta}-\mu_s)^2
\right\}
\times
\int_{\mathbb R}
\exp\left\{
-\frac{\tau+\tau_s}{2}
\left(
\theta-m(\hat{\theta})
\right)^2
\right\}d\theta.
\end{align*}
The remaining integral is a normal integral:
\[
\int_{\mathbb R}
\exp\left\{
-\frac{\tau+\tau_s}{2}
\left(
\theta-m(\hat{\theta})
\right)^2
\right\}d\theta
=
\sqrt{\frac{2\pi}{\tau+\tau_s}}.
\]
Consequently,
\[
O(\hat{\theta})
=
\sqrt{
\frac{\tau\tau_s}
{2\pi(\tau+\tau_s)}
}
\exp\left\{
-\frac{\tau\tau_s}{2(\tau+\tau_s)}
(\hat{\theta}-\mu_s)^2
\right\}.
\]

Thus, the overlap function is uniquely maximized at
\[
\hat{\theta}_c=\mu_s,
\]
which is selected as the compatibility anchor. Therefore,
\[
O_{\max}=O(\hat{\theta}_c).
\]
The corresponding standardized Bayes factor is
\begin{align*}
\mathrm{SBF}(\hat{\theta})
=
\frac{O(\hat{\theta})}
{O_{\max}}
=
\frac{O(\hat{\theta})}
{O(\hat{\theta}_c)}
=
\exp\left\{
-\frac{\tau\tau_s}{2(\tau+\tau_s)}
(\hat{\theta}-\mu_s)^2
\right\}.
\end{align*}
It follows that
\[
w^*(\hat{\theta})
=
\operatorname{logit}^{-1}
\left\{
\operatorname{logit}(w)
-
\frac{\tau\tau_s}{2(\tau+\tau_s)}
(\hat{\theta}-\mu_s)^2
\right\}.
\]

\subsection{Derivation of the Geometric Interpretation of Borrowing Decay}

Under the source model, we have $\theta \mid M_s \sim \Normal(\mu_s, \tau_s^{-1})$ and $\hat{\theta} \mid \theta \sim \Normal(\theta, \tau^{-1})$. Marginalizing over $\theta$ in the normal--normal hierarchy yields the prior predictive distribution $\hat{\theta} \mid M_s \sim \Normal(\mu_s, \tau^{-1} + \tau_s^{-1})$. Consequently, the conflict $\Delta = \hat{\theta} - \mu_s$ follows $\Normal(0, \tau^{-1} + \tau_s^{-1})$. Defining $Z$ as the standardized conflict:
\begin{align*}
Z := \frac{\Delta}{\sqrt{\operatorname{Var}_{M_s}(\hat{\theta})}} = \frac{\Delta}{\sqrt{\tau^{-1} + \tau_s^{-1}}} \sim \Normal(0,1).
\end{align*}
Squaring $Z$ yields:
\begin{align*}
Z^2 = \frac{\Delta^2}{\tau^{-1} + \tau_s^{-1}} = \frac{\tau \tau_s}{\tau + \tau_s} \Delta^2.
\end{align*}
Substituting this into the exponent of the derived standardized Bayes factor $\mathrm{SBF}(\hat{\theta}) = \exp\left\{-\frac{\tau\tau_s}{2(\tau + \tau_s)}\Delta^2\right\}$, we obtain the identity $\mathrm{SBF}(\hat{\theta}) = \exp\left\{-\frac{1}{2}Z^2\right\}$. Since $\frac{w^*(\hat{\theta})/(1-w^*(\hat{\theta}))}{w/(1-w)} = \mathrm{SBF}(\hat{\theta})$, the posterior-to-initial source-component odds ratio decays according to the standard normal kernel in $Z$.

\section{Binomial Distribution Case}

\subsection{Derivation of Component Posteriors}

On the probability scale, the target likelihood is
\[
L(p;x)
\propto
p^x(1-p)^{n-x},
\qquad 0<p<1.
\]

The source prior is
\[
p\mid M_s
\sim
\mathrm{Beta}(\alpha_s,\beta_s),
\]
with density
\[
\pi(p\mid M_s)
=
\frac{
p^{\alpha_s-1}(1-p)^{\beta_s-1}
}{
B(\alpha_s,\beta_s)
}.
\]
Therefore, the source-component posterior satisfies
\begin{align*}
\pi(p\mid M_s,x)
\propto
L(p;x)\pi(p\mid M_s)
\propto
p^{x+\alpha_s-1}
(1-p)^{n-x+\beta_s-1}.
\end{align*}
Consequently,
\[
p\mid M_s,x
\sim
\mathrm{Beta}
\left(
x+\alpha_s,\,
n-x+\beta_s
\right).
\]

To express this posterior on the canonical logit scale, let
\[
\theta=\operatorname{logit}(p),
\qquad
p=\operatorname{logit}^{-1}(\theta),
\qquad
\frac{dp}{d\theta}=p(1-p).
\]
The density with respect to $d\theta$ is therefore
\begin{align*}
\pi(\theta\mid M_s,x)
&=
\pi(p\mid M_s,x)
\left|\frac{dp}{d\theta}\right|\\
&=
\frac{
p^{x+\alpha_s}(1-p)^{n-x+\beta_s}
}{
B(x+\alpha_s,n-x+\beta_s)
}\\
&=
\frac{1}
{B(x+\alpha_s,n-x+\beta_s)}
\exp\left\{
(x+\alpha_s)\theta
-
(n+\alpha_s+\beta_s)\psi(\theta)
\right\},
\end{align*}
where
\[
\psi(\theta)=\log(1+e^\theta).
\]

For the no-borrowing component, the reference prior is flat on the
canonical logit scale:
\[
\pi(\theta\mid M_{nb})\propto1.
\]
On the probability scale, the corresponding prior density is
\[
\pi(p\mid M_{nb})
\propto
\left|\frac{d\theta}{dp}\right|
=
\frac{1}{p(1-p)}.
\]
Combining this reference prior with the target likelihood gives
\begin{align*}
\pi(p\mid M_{nb},x)
\propto
p^x(1-p)^{n-x}
\frac{1}{p(1-p)}
=
p^{x-1}(1-p)^{n-x-1}.
\end{align*}
Hence, for an interior count $0<x<n$,
\[
p\mid M_{nb},x
\sim
\mathrm{Beta}(x,n-x).
\]

Transforming this posterior back to the canonical logit scale gives
\begin{align*}
\pi(\theta\mid M_{nb},x)
=
\pi(p\mid M_{nb},x)
\left|\frac{dp}{d\theta}\right|
=
\frac{
p^x(1-p)^{n-x}
}{
B(x,n-x)
}
=
\frac{1}{B(x,n-x)}
\exp\left\{
x\theta-n\psi(\theta)
\right\}.
\end{align*}

\subsection{Derivation of the Compatibility Anchor and Posterior Weight}

For an interior count $0<x<n$, the no-borrowing component posterior
derived above is
\[
\pi(\theta\mid M_{nb},x)
=
\frac{1}{B(x,n-x)}
\exp\left\{
x\theta-n\psi(\theta)
\right\}.
\]

The source prior is specified on the probability scale as
\[
p\mid M_s
\sim
\mathrm{Beta}(\alpha_s,\beta_s),
\]
with density
\[
f_s(p)
=
\frac{
p^{\alpha_s-1}(1-p)^{\beta_s-1}
}{
B(\alpha_s,\beta_s)
}.
\]
Because
\[
p=\operatorname{logit}^{-1}(\theta)
\qquad\text{and}\qquad
\frac{dp}{d\theta}=p(1-p),
\]
the corresponding source-prior density on the canonical logit scale is
\begin{align*}
\pi(\theta\mid M_s)
=
f_s(p(\theta))
\left|\frac{dp}{d\theta}\right|
=
\frac{
p^{\alpha_s}(1-p)^{\beta_s}
}{
B(\alpha_s,\beta_s)
}
=
\frac{1}{B(\alpha_s,\beta_s)}
\exp\left\{
\alpha_s\theta
-
(\alpha_s+\beta_s)\psi(\theta)
\right\}.
\end{align*}

Therefore, the overlap function is
\begin{align*}
O(x)
&=
\int_{\mathbb R}
\pi(\theta\mid M_{nb},x)
\pi(\theta\mid M_s)\,d\theta\\
&=
\frac{1}
{B(x,n-x)B(\alpha_s,\beta_s)}
\int_{\mathbb R}
\exp\left\{
(x+\alpha_s)\theta
-
(n+\alpha_s+\beta_s)\psi(\theta)
\right\}d\theta.
\end{align*}

The source-component posterior derived in the preceding subsection is
\[
\pi(\theta\mid M_s,x)
=
\frac{1}
{B(x+\alpha_s,n-x+\beta_s)}
\exp\left\{
(x+\alpha_s)\theta
-
(n+\alpha_s+\beta_s)\psi(\theta)
\right\}.
\]
Rearranging this expression gives
\[
\exp\left\{
(x+\alpha_s)\theta
-
(n+\alpha_s+\beta_s)\psi(\theta)
\right\}
=
B(x+\alpha_s,n-x+\beta_s)
\pi(\theta\mid M_s,x).
\]
Substitution into the overlap function yields
\begin{align*}
O(x)
=
\frac{
B(x+\alpha_s,n-x+\beta_s)
}{
B(x,n-x)B(\alpha_s,\beta_s)
}
\int_{\mathbb R}
\pi(\theta\mid M_s,x)\,d\theta
=
\frac{
B(x+\alpha_s,n-x+\beta_s)
}{
B(x,n-x)B(\alpha_s,\beta_s)
},
\end{align*}
because $\pi(\theta\mid M_s,x)$ is a normalized density.

Since $B(\alpha_s,\beta_s)$ does not depend on $x$, the maximizer can be identified by comparing consecutive overlap values. For
$x=1,\ldots,n-2$,
\begin{align*}
\frac{O(x+1)}{O(x)}
&=
\frac{
B(x+1+\alpha_s,n-x-1+\beta_s)
}{
B(x+\alpha_s,n-x+\beta_s)
}
\times
\frac{
B(x,n-x)
}{
B(x+1,n-x-1)
}
=
\frac{
(x+\alpha_s)(n-x-1)
}{
x(n-x+\beta_s-1)
}.
\end{align*}
It follows that
\[
O(x+1)
\begin{cases}
>O(x),&
x<(n-1)\dfrac{\alpha_s}{\alpha_s+\beta_s},\\[6pt]
=O(x),&
x=(n-1)\dfrac{\alpha_s}{\alpha_s+\beta_s},\\[6pt]
<O(x),&
x>(n-1)\dfrac{\alpha_s}{\alpha_s+\beta_s}.
\end{cases}
\]
Therefore, if $(n-1)\tfrac{\alpha_s}{\alpha_s+\beta_s}$ is not an integer, the overlap increases strictly up to
\[
x_c
=
\left\lceil
(n-1)\frac{\alpha_s}{\alpha_s+\beta_s}
\right\rceil
\]
and decreases strictly thereafter. Hence, this $x_c$ is the unique
maximizer.

If $(n-1)\tfrac{\alpha_s}{\alpha_s+\beta_s}$ is an integer,  there are exactly two adjacent
maximizers,
\[
(n-1)\dfrac{\alpha_s}{\alpha_s+\beta_s}
\qquad\text{and}\qquad
(n-1)\dfrac{\alpha_s}{\alpha_s+\beta_s}+1,
\]
and either may be selected as the compatibility anchor.

For any selected compatibility anchor,
\[
O_{\max}=O(x_c).
\]
The standardized Bayes factor is therefore
\begin{align*}
\mathrm{SBF}(x)
=
\frac{O(x)}{O_{\max}}
=
\frac{O(x)}{O(x_c)}
=
\frac{
B(x+\alpha_s,n-x+\beta_s)
}{
B(x,n-x)
}
\times
\frac{
B(x_c,n-x_c)
}{
B(x_c+\alpha_s,n-x_c+\beta_s)
}.
\end{align*}
The corresponding source-component posterior weight is
\[
w^*(x)
=
\operatorname{logit}^{-1}
\left\{
\operatorname{logit}(w)
+
\log\mathrm{SBF}(x)
\right\}.
\]
By construction,
\[
w^*(x)\leq w
\]
for every interior count $x\in\{1,\ldots,n-1\}$, with equality at
$x_c$ and, when the maximum is not unique, at both maximizing counts.

\subsection{Derivation of the Posterior Decision Threshold}
For an interior count $0<x<n$, the no-borrowing posterior satisfies
\[
p\mid M_{nb},x
\sim
\mathrm{Beta}(x,n-x).
\]
Recall that, for positive integers $a$ and $b$, the cumulative distribution function of a $\mathrm{Beta}(a,b)$ distribution at $q$ equals the upper-tail probability at $a$ for a $\mathrm{Binomial}(a+b-1,q)$ distribution. Therefore,
\[
\mathbb P(p\leq p_0\mid M_{nb},x)
=
\mathbb P\{\mathrm{Binomial}(n-1,p_0)\geq x\}.
\]
Taking complements gives
\begin{align*}
\mathbb P(p>p_0\mid M_{nb},x)
=
\mathbb P\{\mathrm{Binomial}(n-1,p_0)\leq x-1\}
=
\sum_{k=0}^{x-1}
\binom{n-1}{k}
p_0^k(1-p_0)^{n-1-k}.
\end{align*}
By the binomial central limit theorem, the cumulative distribution
function of the standardized $\mathrm{Binomial}(n-1,p_0)$ distribution
converges pointwise to the standard normal distribution function
$\Phi$. Since $\Phi$ is continuous, this convergence is uniform. Hence,
\[
\mathbb P(p>p_0\mid M_{nb},x)
=
\Phi\left\{
\frac{x-1-(n-1)p_0}
{\sqrt{(n-1)p_0(1-p_0)}}
\right\}
+o(1),
\]
uniformly over the interior counts $x$.

To evaluate the frequentist Type~I error, let $X$ denote the random target count. Under the null boundary $p=p_0$,
\[
X\sim\mathrm{Binomial}(n,p_0),
\]
and hence, by the binomial central limit theorem,
\[
\frac{X-np_0}
{\sqrt{np_0(1-p_0)}}
\xrightarrow{d}
\Normal(0,1).
\]
Moreover,
\begin{align*}
\frac{X-1-(n-1)p_0}
{\sqrt{(n-1)p_0(1-p_0)}}
=
\sqrt{\frac{n}{n-1}}\,
\frac{X-np_0}
{\sqrt{np_0(1-p_0)}}
-
\frac{1-p_0}
{\sqrt{(n-1)p_0(1-p_0)}}.
\end{align*}
Since
\[
\sqrt{\frac{n}{n-1}}\longrightarrow1
\qquad\text{and}\qquad
\frac{1-p_0}
{\sqrt{(n-1)p_0(1-p_0)}}\longrightarrow0,
\]
Slutsky's theorem gives
\[
\frac{X-1-(n-1)p_0}
{\sqrt{(n-1)p_0(1-p_0)}}
\xrightarrow{d}
\Normal(0,1).
\]

Because the preceding normal approximation is uniform over the interior counts, replacing the fixed count $x$ by the random count $X$ gives, on the event $0<X<n$,
\[
\mathbb P(p>p_0\mid M_{nb},X)
-
\Phi\left\{
\frac{X-1-(n-1)p_0}
{\sqrt{(n-1)p_0(1-p_0)}}
\right\}
\xrightarrow{\mathbb P}0.
\]
Any prespecified convention for the boundary counts can affect the rejection probability only when $X\in\{0,n\}$. Under $p=p_0$,
\[
\mathbb P_{X\mid p_0,n}\{X\in\{0,n\}\}
=
(1-p_0)^n+p_0^n
\longrightarrow0.
\]
Thus, the boundary-count convention does not affect the limiting Type~I error.

Therefore, by Slutsky's theorem and the continuous mapping theorem,
\[
\mathbb P(p>p_0\mid M_{nb},X)
\xrightarrow{d}
\Phi(Z),
\qquad
Z\sim\Normal(0,1).
\]
Because the limiting distribution is continuous at $1-\alpha_0$,
\begin{align*}
\mathbb P_{X\mid p_0,n}
\left\{
\mathbb P(p>p_0\mid M_{nb},X)
\geq1-\alpha_0
\right\}
\longrightarrow
\mathbb P\{\Phi(Z)\geq1-\alpha_0\}
=
\mathbb P\left\{
Z\geq\Phi^{-1}(1-\alpha_0)
\right\}
=
\alpha_0.
\end{align*}
Thus, under no borrowing, the frequentist Type~I error induced by the posterior decision threshold $1-\alpha_0$ converges to $\alpha_0$.
The lower-tail result follows analogously.

\end{document}